\documentclass[11pt]{article}
\usepackage{geometry}                
\usepackage{graphicx}
\usepackage{amssymb}
\usepackage{amsmath}
\usepackage{epstopdf}
\usepackage[title]{appendix}
\usepackage{amsthm}

\newtheorem{proposition}{Proposition}

\newtheorem{conjecture}{Conjecture}
\newtheorem{lemma}{Lemma}[]
\newtheorem{remark}{Remark}[]

\newtheorem{example}{Example}
\usepackage{multirow}
\usepackage{setspace}
\usepackage{authblk}

\usepackage[plain]{algorithm2e}
\usepackage{tikz}

\usepackage[all,cmtip]{xy}
\usepackage{lineno}
\title{Using Fano factors to determine certain types of gene autoregulation}
\author[1,2,*]{Yue Wang}
\author[3]{Siqi He}
\affil[1]{Department of Computational Medicine, University of California, Los Angeles, California, United States of America}
\affil[2]{Institut des Hautes \'Etudes Scientifiques, Bures-sur-Yvette, Essonne, France}
\affil[3]{Simons Center for Geometry and Physics, Stony Brook University, Stony Brook, New York, United States of America}
\affil[*]{E-mail address: yuew@g.ucla.edu (Y. W.). ORCID: 0000-0001-5918-7525}
\date{}                                           
\begin{document}
	\maketitle

\begin{abstract}
	The expression of one gene might be regulated by its corresponding protein, which is called autoregulation. Although gene regulation is a central topic in biology, autoregulation is much less studied. In general, it is extremely difficult to determine the existence of autoregulation with direct biochemical approaches. Nevertheless, some papers have observed that certain types of autoregulations are linked to noise levels in gene expression. We generalize these results by two propositions on discrete-state continuous-time Markov chains. These two propositions form a simple but robust method to infer the existence of autoregulation in certain scenarios from gene expression data. This method only depends on the Fano factor, namely the ratio of variance and mean of the gene expression level. Compared to other methods for inferring autoregulation, our method only requires non-interventional one-time data, and does not need to estimate parameters. Besides, our method has few restrictions on the model. We apply this method to four groups of experimental data and find some genes that might have autoregulation. Some inferred autoregulations have been verified by experiments or other theoretical works.
\end{abstract}

\smallskip
\noindent \textbf{Keywords.} 

\noindent inference; gene expression; autoregulation; Markov chain.

\

\noindent \textbf{Frequently used abbreviations:} 

\noindent GRN: gene regulatory network.

\noindent VMR: variance-to-mean ratio

\section{Introduction}
\label{intro}

In general, genes are transcribed to mRNAs and then translated to proteins. We can use the abundance of mRNA or protein to represent the expression levels of genes. Both the synthesis and degradation of mRNAs/proteins can be affected (activated or inhibited) by the expression levels of other genes \cite{karamyshev2018lost}, which is called (mutual) gene regulation. Genes and their regulatory relations form a gene regulatory network (GRN) \cite{cunningham2015mechanisms}, generally represented as a directed graph: each vertex is a gene, and each directed edge is a regulatory relationship. See Fig.~\ref{grn} for an example of a GRN.

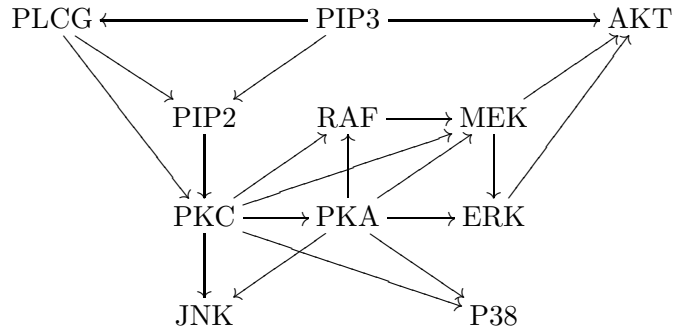
\begin{figure}
	\center
	$\xymatrix{
		\text{PLCG}\ar[rd]\ar[rdd]&&\text{PIP3}\ar[ll]\ar[ld]\ar[rr]&&\text{AKT}\\
		&\text{PIP2}\ar[d]&\text{RAF}\ar[r]&\text{MEK}\ar[ru]\ar[d]&\\
		&\text{PKC}\ar[ru]\ar[rru]\ar[r]\ar[d]\ar[rrd]&\text{PKA}\ar[u]\ar[ru]\ar[r]\ar[rd]\ar[ld]&\text{ERK}\ar[ruu]&\\
		&\text{JNK}&&\text{P38}&
	}$\\
	\caption{An example of a GRN in human T cells \cite{werhli2006comparative}. Each vertex is a gene. Each arrow is a regulatory relationship. Notice that it has no directed cycle.}
	\label{grn}
\end{figure}

The expression of one gene could promote/repress its own expression, which is called positive/negative autoregulation \cite{carrier1999investigating}. Autoregulation is very common in \emph{E. coli} \cite{shen2002network}. Positive autoregulation is also called autocatalysis or autoactivation, and negative autoregulation is also called autorepression \cite{baumdick2018conformational,fang2017sirt7}. For instance, HOX proteins form and maintain spatially inhomogeneous expression of HOX genes \cite{sheth2014self}. For genes with position-specific expressions during development, it is common that the increase of one gene can further increase or decrease its level \cite{wang2020biological}. Autoregulation has the effect of stabilizing transposons in genomes \cite{bouuaert2013autoregulation}, which can affect cell behavior \cite{kang2015flexibility,wang2023longest}. Autoregulation can also stabilize the cell phenotype \cite{barros2011cdx2}, which is related to cancer development \cite{zhou2014multi,niu2015phenotypic,chen2016overshoot,wang2023multiple}.

While countless works infer the regulatory relationships between different genes (the GRN structure) \cite{wang2022inference}, determining the existence of autoregulation is an equally important yet less-studied field. Due to technical limitations, it is difficult and sometimes impossible to directly detect autoregulation in experiments. Instead, we can measure gene expression profiles and infer the existence of autoregulation. In this paper, we consider a specific data type: measure the expression levels of certain genes without intervention for a single cell (which reaches stationarity) at a single time point, and repeat for many different cells to obtain a probability distribution for expression levels. Such single-cell non-interventional one-time gene expression data can be obtained with a relatively low cost \cite{luecken2019current}.

With such single-cell level data for one gene $V$, we can calculate the ratio of variance and mean of the expression level (mRNA or protein count). This quantity is called the variance-to-mean ratio (VMR) or the Fano factor. Many papers that study gene expression systems with autoregulations have found that negative autoregulation can decrease noise (smaller VMR), and positive autoregulation can increase noise (larger VMR) \cite{thattai2001intrinsic,swain2004efficient,hornos2005self,munsky2012using,gronlund2013transcription,dessalles2017stochastic,czuppon2018limits}. This means VMR can be used to infer the existence of autoregulation. 

We generalize the above observation and develop two mathematical results that use VMR to determine the existence of autoregulation. They apply to some genes that have autoregulation. For genes without autoregulation, these results cannot determine that autoregulation does not exist. We apply these results to four experimental gene expression data sets and detect some genes that might have autoregulation. 

We start with some setup and introduce our main results (Section~\ref{setup}). Then we cite some previous works on this topic and compare them with our results (Section~\ref{related}). For a single gene that is not regulated by other genes (Section~\ref{auto}) and multiple genes that regulate each other (Section~\ref{multi}), we develop mathematical results to identify the existence of autoregulation. These two mathematical sections can be skipped. We summarize the procedure of our method and apply it to experimental data (Section~\ref{app}). We finish with some conclusions and discussions (Section~\ref{con}).

\section{Setup and main results}
\label{setup}
One possible mechanism of ``the increase of one gene's expression level further increases its expression level'' is a positive feedback loop between two genes \cite{hui2020increased}. Here $V_1$ and $V_2$ promote each other, so that the increase of $V_1$ increases $V_2$, which in return further increases $V_1$. We should not regard this feedback loop as autoregulation. When we define autoregulation for a gene $V$, we should fix environmental factors and other genes that regulate $V$, and observe whether the expression level of $V$ can affect itself. If $V$ is in a feedback loop that contains other genes, then those genes (which regulate $V$ and are regulated by $V$) cannot be fixed when we change $V$. Therefore, it is essentially difficult to determine whether $V$ has autoregulation in this scenario. In the following, we need to assume that $V$ is not contained in a feedback loop that involves other genes.

The actual gene expression mechanism might be complicated. Besides other genes/factors that can regulate a gene, for a gene $V$ itself, it might switch between inactivated (off) and activated (on) states \cite{cao2020analytical}. These states correspond to different transcription rates to produce mRNAs. When mRNAs are translated into proteins, those proteins might affect the transition of gene activation states, which forms autoregulation \cite{firman2018maximum}. See Fig.~\ref{mech} for an illustration. Therefore, for a gene $V$, we should regard the gene activation state, mRNA count, and protein count as a triplet of random variables $(G,M,P)$, which depend on each other. 

\begin{figure}
	\center
	$\xymatrix{
		\text{inactivated gene}\ar[rd]\ar@/^/[dd]&&\\
		&\text{mRNA}\ar[r]&\text{protein}\ar@{-->}@/_2pc/[ll]\\
		\text{activated gene}\ar[ru]\ar@/^/[uu]&
	}$\\
	\caption{The mechanism of gene expression. A gene might switch between inactivated state and activated state, which correspond to different transcription rates. Gene is transcribed into mRNAs, which are translated into proteins. Proteins might (auto)regulate the state transition of the corresponding gene.}
	\label{mech}
\end{figure}
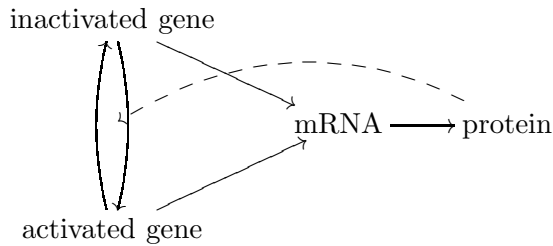

When we fix environmental factors and other genes that affect $V$, the triplet $(G,M,P)$ should follow a continuous-time Markov chain. A possible state is the gene activation state on/off (for $G$), the mRNA count on $\mathbb{Z}$ (for $M$), and the protein count on $\mathbb{Z}$ (for $P$). Thus the total space is $\{0,1\}\times \mathbb{Z} \times \mathbb{Z}$. When we consider the expression level $M$ or $P$ (but have no access to the value of $G$), sometimes itself is Markovian (its dynamics can be fully determined by itself, without the knowledge of $G$), and we call this scenario ``\textbf{autonomous}''. In other cases, $M$ or $P$ itself is no longer Markovian (its dynamics explicitly depends on $G$), and we call this scenario ``\textbf{non-autonomous}''. We need to consider the triplet $(G,M,P)$ in the non-autonomous scenario. This is similar to a hidden Markov model, where a two-dimensional Markov chain is no longer Markovian if projected to one dimension (since this dimension depends on the other dimension).

For the autonomous scenario, we can fully classify autoregulation for a gene $V$. Assume environmental factors and other genes that affect the expression of $V$ are kept at constants. Define the expression level (mRNA count for example) of one cell to be $X=n$, the mRNA synthesis rate at $X=n-1$ to be $f_n$, and the degradation rate for each mRNA molecule at $X=n$ to be $g_n$. This is a standard continuous-time Markov chain on $\mathbb{Z}$ with transition rates 
\[\frac{1}{\Delta t}\mathbb{P}[X(t+\Delta t)=n\mid X(t)=n-1]=f_n,\]
\[\frac{1}{\Delta t}\mathbb{P}[X(t+\Delta t)=n-1\mid X(t)=n]=ng_n.\]
Define the relative growth rate $h_n=f_n/g_n$. If there is \textbf{no autoregulation}, then $h_n$ is a constant. \textbf{Positive autoregulation} means $h_n>h_{n-1}$ for some $n$, so that $f_n>f_{n-1}$ and/or $g_n<g_{n-1}$; \textbf{negative autoregulation} means $h_n<h_{n-1}$ for some $n$, so that $f_n<f_{n-1}$ and/or $g_n>g_{n-1}$. Notice that we can have $h_n>h_{n-1}$ for some $n$ and $h_{n'}<h_{n'-1}$ for some other $n'$, meaning that positive autoregulation and negative autoregulation can both exist for the same gene, but occur at different expression levels. Such non-monotonicity in regulating gene expression often appear in reality \cite{angelini2022model}.

For the non-autonomous scenario, we can still define autoregulation. Consider the expression level $X$ of $V$ (mRNA count or protein count) and its interior factor $I$. If $X$ is the mRNA count, then $I$ is the gene state; if $X$ is the protein count, then $I$ is the gene state and the mRNA count. If there is \textbf{no autoregulation}, then $X$ cannot affect $I$, and for each value of $I$, the relative growth rate $h_n$ of $X$ is a constant. If $X$ can affect $I$, or $h_n$ is not a constant, then there is \textbf{autoregulation}. When $X$ can affect $I$, there is a directed cycle ($X\to I\to X$), and the change of $X$ can affect itself through $I$. In this case, it is not always easy to distinguish between positive autoregulation and negative autoregulation.

Quantitatively, for the autonomous scenario, when we fix other factors that might regulate this gene $V$, if $V$ has no autoregulation, then $h_n=f_n/g_n$ is a constant $h$ for all $n$. In this case, the stationary distribution of $V$ satisfies $\mathbb{P}(X=n)/\mathbb{P}(X=n-1)=h/n$, meaning that the distribution is Poissonian with parameter $h$, $\mathbb{P}(X=n)=h^ne^{-h}/n!$, and $\text{VMR}=1$. If there exists positive autoregulation of certain forms, $\text{VMR}>1$; if there exists negative autoregulation of certain forms, $\text{VMR}<1$. However, such results are derived by assuming that $f_n,g_n$ take certain functional forms, such as linear functions \cite{paulsson2005models,ramos2015gene}, quadratic functions \cite{giovanini2020comparative}, or Hill functions \cite{stewart2013under}. There are other papers that consider Markov chain models in gene expression/regulation \cite{jia2017simplification,sharma2014markov,shmulevich2003steady,chen2020limit,shen2019distributed,ko2019markov}, but the role of VMR is not thoroughly studied.

In this paper, we generalize the above result of inferring autoregulation with VMR by dropping the restrictions on parameters. Consider a gene $V$ in a known GRN, and assume it is not regulated by other genes, or assume other factors that regulate $V$ are fixed. Assume we have the \textbf{autonomous} scenario, meaning that its expression level $X=n$ satisfies a general Markov chain with synthesis rate $f_n$ and per molecule degradation rate $g_n$. We do not add any restrictions on $f_n$ and $g_n$. Use the single-cell non-interventional one-time gene expression data to calculate the VMR of $V$. Proposition~\ref{prop2} states that $\text{VMR}>1$ or $\text{VMR}<1$ means the existence of positive/negative autoregulation.

Nevertheless, the autonomous condition requires some assumptions, and often does not hold in reality \cite{bokes2012multiscale,jia2017emergent,jia2020kinetic,jia2017simplification}. Consider a gene $V$ that is not regulated by other genes, and has no autoregulation. The mRNA count or the protein count is regulated by the gene activation state (an interior factor), which cannot be fixed. Due to this non-controllable factor, there might be transcriptional bursting \cite{shahrezaei2008analytical,dobrinic2021prc1} or translational bursting \cite{cagnetta2019noncanonical}, where transcription or translation can occur in bursts, and we have $\text{VMR}>1$. This does not mean that Proposition~\ref{prop2} is wrong. Instead, it means that the expression level itself is not Markovian, and the scenario is non-autonomous. In this scenario, we should apply Proposition~\ref{np}, described below, which states that no autoregulation means $\text{VMR}\ge 1$.

We extend the idea of inferring autoregulation with VMR to a gene that is regulated by other genes, or with non-autonomous expression. Consider a gene $V'$ in a known GRN. Denote other genes that regulate $V'$ and the interior factors (gene state and/or mRNA count) of $V'$ by $\boldsymbol{F}$. Denote the values of $V',\boldsymbol{F}$ as $X,\boldsymbol{Y}$. Assume $V'$ is not contained in a feedback loop, and assume $g_n$, the per molecule degradation rate of $V'$ at $X=n$, is not regulated by other genes or its interior factors (gene state and/or mRNA count). We do not add any restrictions on the synthesis rate $f_n$. Proposition~\ref{np} states that if $V'$ has no autoregulation, then $\text{VMR}(X)\ge 1$. Therefore, $\text{VMR}(X)<1$ means autoregulation for $V'$. 

Proposition~\ref{np} is derived in a ``one-step'' Markov chain model, where at one time point, only transitions to the nearest neighbors are allowed: $(X=n,\boldsymbol{Y}=\boldsymbol{a})\to (X=n+1,\boldsymbol{Y}=\boldsymbol{a})$, $(X=n,\boldsymbol{Y}=\boldsymbol{a})\to (X=n-1,\boldsymbol{Y}=\boldsymbol{a})$, and $(X=n,\boldsymbol{Y}=\boldsymbol{a})\to (X=n,\boldsymbol{Y}=\boldsymbol{a}')$. This one-step Markov chain model is the most common approach in stochastic representations of gene regulation \cite{thattai2001intrinsic,hornos2005self,paulsson2005models,munsky2012using,czuppon2018limits}. Recently, there are some studies that consider ``multi-step'' Markov chain models, where at one time point, the change of mRNA/protein count can be accompanied with the change of other factors, such as the gene state \cite{braichenko2021distinguishing,karmakar2021effect,voliotis2008fluctuations}. For example, the following transition is allowed: $(G^*,M=n)\to (G,M=n+1)$. In this multi-step model, Proposition~\ref{np} is no longer valid: even without autoregulation, it is possible that $\text{VMR}(X)<1$. Consider an example that the production of one mRNA molecule needs many steps of gene state transition, and the gene returns to the initial step after producing one mRNA molecule: $G_1\to G_2\to\cdots\to G_k\to G_1+M$, $M\to \emptyset$. Since there are many steps, the total time for one cycle of $G_1\to\cdots\to G_k\to G_1+M$ can be highly deterministic, such as $1$ second. Assume the degradation probability for each mRNA molecule in $1$ second is $0.01$. Then the mRNA count is highly concentrated near $100$, and $\text{VMR}(X)<1$ (close to $0.5$ in numerical simulations).

Since multi-step models allow more transitions, they are more general than one-step models. However, it is still a question that whether such generalizations are necessary, since one-step models have good fitting with experimental data \cite{jia2017stochastic,dessalles2017stochastic,cao2018linear}. Proposition~\ref{np} provides a method to verify this problem: If a gene has $\text{VMR}(X)<1$, but we use other methods to determine that it has no autoregulation, then Proposition~\ref{np} states that one-step models deviate from reality, and multi-step models should be adopted. Therefore, when one-step models hold, Proposition~\ref{np} is a valid method to determine the existence of autoregulation; when one-step models do not hold, combined with other methods to determine autoregulation, Proposition~\ref{np} can detect the failure of one-step models.

In the scenario that Proposition~\ref{np} may apply, if $\text{VMR}\ge 1$, Proposition~\ref{np} cannot determine whether autoregulation exists. In fact, with VMR, or even the full probability distribution, we might not distinguish a non-autonomous system with autoregulation from a non-autonomous system without autoregulation, which both have $\text{VMR}\ge 1$ \cite{cao2018linear}. In the non-autonomous scenario, we only focus on the less complicated case of $\text{VMR}<1$, and derive Proposition~\ref{np} that firmly links VMR and autoregulation.

In reality, Proposition~\ref{prop2} and Proposition~\ref{np} can only apply to a few genes (which are not regulated by other genes or have $\text{VMR}<1$), and they cannot determine negative results. Thus the inference results about autoregulation are a few ``yes'' and many ``we do not know''. Besides, for the results inferred by Proposition~\ref{prop2}, especially those with $\text{VMR}>1$ (positive autoregulation), we cannot verify whether their expression is autonomous, and the inference results are less reliable. 

Current experimental methods can hardly determine the existence of autoregulation, and to determine that a gene does not have autoregulation is even more difficult. Therefore, about whether genes in a GRN have autoregulation, experimentally, we do not have ``yes'' or ``no'', but a few ``yes'' and many ``we do not know''. Thus there is no gold standard to thoroughly evaluate the performance of our inference results. We can only report that some genes inferred by our method to have autoregulation are also verified by experiments or other inference methods to have autoregulation. Besides, if the result by Proposition~\ref{np} does not match with other methods, it is possible that the one-step model fails. Instead, in Section~\ref{sim}, we test our methods with numerical simulations, and the performances of both Propositions are satisfactory.

\section{Related works}
\label{related}

There have been some results of inferring  autoregulation with VMR \cite{thattai2001intrinsic,swain2004efficient,hornos2005self,munsky2012using,gronlund2013transcription,dessalles2017stochastic,czuppon2018limits}. However, these VMR-based methods have various restrictions on the model, and some of them are derived by applying linear noise approximations, which are not always reliable in gene regulatory networks \cite{thomas2013reliable}.

Besides VMR-based methods, there are other mathematical approaches to infer the existence of autoregulation in gene expression \cite{sanchez2018bayesian,xing2005causal,feigelman2016analysis,veerman2021parameter,jia2018relaxation,zhou2012analytical,jia2020small,jia2020dynamical}. We introduce some works and compare them with our method. (\textbf{A}) Sanchez-Castillo et al. \cite{sanchez2018bayesian} considered an autoregressive model for multiple genes. This method (1) needs time series data; (2) requires the dynamics to be linear; (3) estimates a group of parameters. (\textbf{B}) Xing et al. \cite{xing2005causal} applied causal inference to a complicated gene expression model. This method (1) needs promoter sequences and information on transcription factor binding sites; (2) requires linearity for certain steps; (3) estimates a group of parameters. (\textbf{C}) Feigelman et al. \cite{feigelman2016analysis} applied a Bayesian method for model selection. This method (1) needs time series data; (2) estimates a group of parameters. (\textbf{D}) Veerman et al. \cite{veerman2021parameter} considered the probability-generating function of a propagator model. This method (1) needs time series data; (2) estimates a group of parameters; (3) needs to approximate a Cauchy integral. (\textbf{E}) Jia et al. \cite{jia2018relaxation} compared the relaxation rate with degradation rate. This method (1) needs interventional data; (2) only works for a single gene that is not regulated by other genes; (3) requires that the per molecule degradation rate is a constant.

Compared to other more complicated methods, VMR-based methods (including ours) have two advantages: (1) VMR-based methods use non-interventional one-time data. Time series data require measuring the same cell multiple times without killing it, and interventional data require some techniques to interfere with gene expression, such as gene knockdown. Therefore, non-interventional one-time data used in VMR-based methods are much easier and cheaper to obtain. (2) VMR-based methods do not estimate parameters, and only calculate the mean and variance of the expression level. Some other methods need to estimate many parameters or approximate some complicated quantities, meaning that they need large data size and high data accuracy. Therefore, our method is easy to calculate, and need lower data accuracy and smaller data size. 

Compared to other VMR-based methods, our method has few restrictions on the model, making them applicable to various scenarios with different dynamics. Besides, our derivations do not use any approximations. 

In sum, compared to other VMR-based methods, our method is universal. Compared to other non-VMR-based methods, our method is simple, and has lower requirements on data quality.

Compared to other non-VMR-based methods, our method has some disadvantages: (1) The GRN structure needs to be known. (2) Our method does not work for certain genes, depending on regulatory relationships. Proposition~\ref{prop2} only works for a gene that is not regulated by other genes, and we require its expression to be autonomous; Proposition~\ref{np} only works for a gene that is not in a feedback loop. (3) Proposition~\ref{np} requires the per molecule degradation rate to be a constant, and it cannot provide information about autoregulation if $\text{VMR}\ge 1$. (4) Our method only works for cells at equilibrium. Thus time series data that contain time-specific information cannot be utilized other than treated as one-time data. With just the stationary distribution, sometimes it is impossible to build the causal relationship (including autoregulation) \cite{wang2020causal}. Thus with this data type, some disadvantages are inevitable. Such impossibility results might be generalized to other data types or even other fields \cite{wang2022impossibility}.

\section{Scenario of a single isolated gene}
\label{auto}

\subsection{Setup}
We first consider the expression level (e.g., mRNA count) of one gene $V$ in a single cell. At the single-cell level, gene expression is essentially stochastic, and we do not further consider differential equation approaches \cite{wang2022modelling} dynamical system approaches with deterministic \cite{wang2020model} or stochastic \cite{ye2016stochastic} operators. We use a random variable $X$ to represent the mRNA count of $V$. We assume $V$ is not in a feedback loop. We also assume all environmental factors and other genes that can affect $X$ are kept at constant levels, so that we can focus on $V$ alone. This can be achieved if no other genes point to gene $V$ in the GRN, such as PIP3 in Fig.~\ref{grn}. Then we assume that the expression of $V$ is autonomous, thus $X$ satisfies a time-homogeneous Markov chain defined on $\mathbb{Z}^*$. 

Assume that the mRNA synthesis rate at $X(t)=n-1$, namely the transition rate from $X=n-1$ to $X=n$, is $f_n\ge 0$. Assume that with $n$ mRNA molecules, the degradation rate for each mRNA molecule is $g_n>0$. Then the overall degradation rate at $X(t)=n$, namely the transition rate from $X=n$ to $X=n-1$, is $g_nn$. The associated master equation is 
\begin{equation}
	\begin{split}
		\frac{\mathrm{d}\mathbb{P}[X(t)=n]}{\mathrm{d}t}=&\mathbb{P}[X(t)=n+1]g_{n+1}(n+1)+\mathbb{P}[X(t)=n-1]f_n\\
		&-\mathbb{P}[X(t)=n](f_{n+1}+g_nn).
	\end{split}
	\label{eq0}
\end{equation}
When $f_n,g_n$ take specific forms, this master equation also corresponds to a branching process, so that related techniques can be applied \cite{jiang2017phenotypic}. Define the relative growth rate $h_n=f_n/g_n$.  We assume that as time tends to infinity, the process reaches equilibrium, where (1) the stationary probability distribution $P_n=\lim_{t\to\infty}\mathbb{P}[X(t)=n]$ exists, and $P_n=P_{n-1}h_n/n$; (2) the mean $\lim_{t\to\infty}\mathbb{E}[X(t)]$ and the variance $\lim_{t\to\infty}\sigma^2[X(t)]$ are finite. Such requirements can be satisfied under simple assumptions, such as assuming $h_n$ has a finite upper bound \cite{norris1998markov,wang2022discrete}.

If $h_n>h_{n-1}$ for some $n$, then there exists positive autoregulation. If $h_n<h_{n-1}$ for some $n$, then there exists negative autoregulation. If there is no autoregulation, then  $h_n$ is a constant $h$, and the stationary distribution is Poissonian with parameter $h$. In this setting, positive autoregulation and negative autoregulation might coexist, meaning that $h_{n+1}<h_n$ for some $n$ and $h_{n'+1}>h_{n'}$ for some $n'$.

\subsection{Theoretical results}

With single-cell non-interventional one-time gene expression data for one gene, we have the stationary distribution of the Markov chain $X$. We can infer the existence of autoregulation with the VMR of $X$, defined as $\text{VMR}(X)=\sigma^2(X)/\mathbb{E}(X)$. The idea is that if we let $f_n$ increase/decrease with $n$, and control $g_n$ to make $\mathbb{E}(X)$ invariant, then the variance $\sigma^2(X)$ increases/decreases \cite[Section 2.5.1]{wang2018some}. We shall prove that $\text{VMR}>1$ implies the occurrence of positive autoregulation, and $\text{VMR}<1$ implies the occurrence of negative autoregulation. Notice that $\text{VMR}>1$ does not exclude the possibility that negative autoregulation exists for some expression level. This also applies to $\text{VMR}<1$ and positive autoregulation.

We can illustrate this result with a linear model: 
\begin{example}
	Consider a Markov chain that satisfies Eq.~\ref{eq0}, and set $f_n=k+b(n-1)$, $g_n=c$. Here $b$ (can be positive or negative) is the strength of autoregulation, and $c$ satisfies $c>0$ and $c-b>0$. We can calculate that $\text{VMR}=1+b/(c-b)$ (see Appendix~\ref{a0.1} for details). Therefore, $\text{VMR}>1$ means positive autoregulation, $b>0$; $\text{VMR}<1$ means negative autoregulation, $b<0$; $\text{VMR}=1$ means no autoregulation, $b=0$.
	\label{ex1}
\end{example}

\begin{lemma}
	Consider a Markov chain $X(t)$ that follows Eq.~\ref{eq0} with general transition coefficients $f_n,g_n$. Here $X(t)$ models the mRNA/protein count of one gene whose expression is autonomous. Calculate $\text{VMR}(X)$ at stationarity. (1) Assume $h_{n+1}\ge h_n$ for all $n$. We have $\text{VMR}(X)\ge 1$; moreover, $\text{VMR}(X)= 1$ if and only if $h_{n+1}= h_n$ for all $n$. (2) Assume $h_{n+1}\le h_n$ for all $n$. We have $\text{VMR}(X)\le 1$; moreover, $\text{VMR}(X)= 1$ if and only if $h_{n+1}= h_n$ for all $n$. 
	\label{lemma0}
\end{lemma}
We can take negation of Lemma~\ref{lemma0} to obtain the following proposition.
\begin{proposition}
	\label{prop2}	
	In the setting of Lemma~\ref{lemma0}, (1) If $\text{VMR}(X)>1$, then there exists at least one value of $n$ for which $h_{n+1}>h_n$; thus this gene has positive autoregulation. (2) If $\text{VMR}(X)<1$, then there exists at least one value of $n$ for which $h_{n+1}<h_n$; thus this gene has negative autoregulation. (3) If $\text{VMR}(X)=1$, then either (A) $h_{n+1}=h_n$ for all $n$, meaning that this gene has no autoregulation; or (B) $h_{n+1}<h_n$ for one $n$ and $h_{n'+1}>h_{n'}$ for another $n'$, meaning that this gene has both positive and negative autoregulation (at different expression levels).
\end{proposition}

\begin{remark}
Results similar to Proposition~\ref{prop2} have been proven by Jia et al. in another model of expression for a single gene \cite{jia2017stochastic}. However, they require that $g_i=g_j$ for any $i,j$. Proposition~\ref{prop2} can handle arbitrary $g_i$, thus being novel.
\end{remark}
\begin{proof}[Proof of Lemma~\ref{lemma0}]
	Define $\lambda=-\log P_0$, so that $P_0=\exp(-\lambda)$. Define $d_n=\prod_{i=1}^{n}h_i>0$ and stipulate that $d_0=1$. We can see that 
	\[\frac{d_nd_{n+2}}{d_{n+1}^2}=\frac{h_{n+2}}{h_{n+1}}.\]
	Also, 
	\[P_n=P_{n-1}f_n/(g_n n)=P_{n-1}h_n/n=\cdots=P_0(\prod_{i=1}^{n}h_i)/n!\ =e^{-\lambda}\frac{d_n}{n!}.\]
	Then 
	\begin{equation*}
		\begin{split}
			\mathbb{E}(X^2)-\mathbb{E}(X)&=\sum_{n=1}^{\infty}(n^2-n)P_n=e^{-\lambda}\sum_{n=1}^{\infty}(n^2-n)\frac{d_n}{n!}\\
			&=e^{-\lambda}\sum_{n=2}^{\infty}\frac{d_n}{(n-2)!}=e^{-\lambda}\sum_{n=0}^{\infty}\frac{d_{n+2}}{n!},
		\end{split}
	\end{equation*}
	\[[\mathbb{E}(X)]^2=\left(\sum_{n=1}^{\infty}nP_n\right)^2=e^{-2\lambda}\left(\sum_{n=1}^{\infty}n\frac{d_n}{n!}\right)^2=e^{-2\lambda}\left(\sum_{n=0}^{\infty}\frac{d_{n+1}}{n!}\right)^2.\]
	Besides,
	\[1=\sum_{n=0}^{\infty}P_n=e^{-\lambda}\sum_{n=0}^{\infty}\frac{d_n}{n!}.\]
	Now we have 
	\[\mathbb{E}(X^2)-\mathbb{E}(X)-[\mathbb{E}(X)]^2=e^{-2\lambda}\left(\sum_{n=0}^{\infty}\frac{d_n}{n!}\right)\left(\sum_{n=0}^{\infty}\frac{d_{n+2}}{n!}\right)-e^{-2\lambda}\left(\sum_{n=0}^{\infty}\frac{d_{n+1}}{n!}\right)^2.\]
	
	(1) Assume $h_{n+1}\ge h_n$ for all $n$. Then 
	
	\begin{equation}
		\label{eqe}
		\begin{split}
			&\mathbb{E}(X^2)-\mathbb{E}(X)-[\mathbb{E}(X)]^2\\
			\ge & e^{-2\lambda}\left(\sum_{n=0}^{\infty}\frac{\sqrt{d_nd_{n+2}}}{n!}\right)^2- e^{-2\lambda}\left(\sum_{n=0}^{\infty}\frac{d_{n+1}}{n!}\right)^2\ge 0.
		\end{split}
	\end{equation}
	Here the first inequality is from the Cauchy-Schwarz inequality, and the second inequality is from $d_nd_{n+2}\ge d_{n+1}^2$ for all $n$. Then $\text{VMR}(X)=\{\mathbb{E}(X^2)-[\mathbb{E}(X)]^2\}/\mathbb{E}(X)\ge 1$. Equality holds if and only if $d_n/d_{n+2}=d_{n+1}/d_{n+3}$ for all $n$ (the first inequality of Eq.~\ref{eqe}) and $d_nd_{n+2}=d_{n+1}^2$ for all $n$ (the second inequality of Eq.~\ref{eqe}). The equality condition is equivalent to $h_{n+1}=h_n$ for all $n$.
	
	(2) Assume $h_{n+1}\le h_n$ for all $n$. Then $d_{n+2}/d_{n+1}\le d_{n+1}/d_n$, and $d_n\le h_1^n$ for all $n$. Define 
	\[H(t)=\sum_{n=0}^{\infty}\frac{d_n}{n!}t^n.\]
	Since $0<d_n\le h_1^n$, this series converges for all $t\in\mathbb{C}$, so that $H(t)$ is a well-defined analytical function on $\mathbb{C}$, and 
	\[H'(t)=\sum_{n=0}^{\infty}\frac{d_{n+1}}{n!}t^n,\ \text{ and }\ H''(t)=\sum_{n=0}^{\infty}\frac{d_{n+2}}{n!}t^n.\]
	In the following, we only consider $H(t),H'(t),H''(t)$ as real functions for $t\in\mathbb{R}$.
	
	To prove $\text{VMR}(X)\le 1$, we just need to prove $\mathbb{E}(X^2)-\mathbb{E}(X)-[\mathbb{E}(X)]^2=e^{-2\lambda}\{H(1)H''(1)-[H'(1)]^2\}\le 0$. However, we shall prove $H''(t)H(t)\le [H'(t)]^2$ for all $t\in \mathfrak{I}$, where $\mathfrak{I}=(a,b)$ is a fixed interval in $\mathbb{R}$ with $0<a<1$ and $1<b<\infty$. Thus $t=1$ is an interior point of $\mathfrak{I}$. Since $H(t),H'(t),H''(t)$ have positive lower bounds on $\mathfrak{I}$, the following statements are obviously equivalent: (i) $H''(t)H(t)\le [H'(t)]^2$ for all $t\in\mathfrak{I}$; (ii) $\{\log[H'(t)/H(t)]\}'\le 0$ for all $t\in\mathfrak{I}$; (iii) $\log[H'(t)/H(t)]$ is non-increasing on $\mathfrak{I}$; (iv) $H'(t)/H(t)$ is non-increasing on $\mathfrak{I}$. To prove (i), we just need to prove (iv).
	
	Consider any $t_1,t_2\in \mathfrak{I}$ with $t_1\le t_2$ and any $p,q\in \mathbb{N}$ with $p\ge q$. Since $d_{p+1}/d_p\le d_{q+1}/d_q$, and $t_1^{p-q}\le t_2^{p-q}$, we have 
	\[
	d_pd_qt_1^qt_2^q(\frac{d_{p+1}}{d_p}-\frac{d_{q+1}}{d_q})(t_1^{p-q}-t_2^{p-q})\ge 0,
	\]
	which means 
	\[
	d_{p+1}d_qt_1^pt_2^q+d_{q+1}d_p t_1^qt_2^p\geq d_{p+1}d_qt_2^pt_1^q+d_{q+1}d_pt_2^qt_1^p.
	\]
	Sum over all $p,q\in \mathbb{N}$ with $p\ge q$ to obtain
	\begin{equation*}
		\begin{split}
			H'(t_1)H(t_2)&=(\sum_{n=0}^{\infty}\frac{d_{n+1}}{n!}t_1^n)(\sum_{n=0}^{\infty}\frac{d_n}{n!}t_2^n)\\
			&\ge (\sum_{n=0}^{\infty}\frac{d_{n+1}}{n!}t_2^n)(\sum_{n=0}^{\infty}\frac{d_n}{n!}t_1^n)=H'(t_2)H(t_1).
		\end{split}
	\end{equation*}
	Thus $H'(t_1)/H(t_1)\ge H'(t_2)/H(t_2)$ for all $t_1,t_2\in \mathfrak{I}$ with $t_1\le t_2$. This means $H''(t)H(t)\le [H'(t)]^2$ for all $t\in\mathfrak{I}$, and $\text{VMR}(X)\le 1$.
	
	About the condition for the equality to hold, assume $h_{n'+1}<h_{n'}$ for a given $n'$. Then 
	\[
	d_{n'}d_{n'-1}t_1^{n'-1}t_2^{n'-1}(\frac{d_{n'+1}}{d_{n'}}-\frac{d_{n'}}{d_{n'-1}})(t_1-t_2)\ge C(t_2-t_1)
	\]
	for all $t_1,t_2\in \mathfrak{I}$ with $t_1\le t_2$ and a constant $C$ that does not depend on $t_1,t_2$. Therefore, 
	\begin{equation*}
		\begin{split}
			&[H'(t_1)/H(t_1)-H'(t_2)/H(t_2)]\cdot[H(t_1)H(t_2)]\\
			=&(\sum_{n=0}^{\infty}\frac{d_{n+1}}{n!}t_1^n)(\sum_{n=0}^{\infty}\frac{d_n}{n!}t_2^n)- (\sum_{n=0}^{\infty}\frac{d_{n+1}}{n!}t_2^n)(\sum_{n=0}^{\infty}\frac{d_n}{n!}t_1^n)\\
			\ge &d_{n'}d_{n'-1}t_1^{n'-1}t_2^{n'-1}(\frac{d_{n'+1}}{d_{n'}}-\frac{d_{n'}}{d_{n'-1}})(t_1-t_2)\\
			\ge& C(t_2-t_1).
		\end{split}
	\end{equation*}
	Since $H(t)$ has a finite positive upper bound $A$ and a positive lower bound $B$ on $\mathfrak{I}$, we have \[H'(t_1)/H(t_1)-H'(t_2)/H(t_2)\ge C(t_2-t_1)/A^2,\] 
	meaning that 
	\[\forall t\in\mathfrak{I},\,\,[H'(t)/H(t)]'=\{H(t)H''(t)-[H'(t)]^2\}/[H(t)]^2\le -C/A^2,\]
	and thus 
	\[\forall t\in\mathfrak{I},\,\,H(t)H''(t)-[H'(t)]^2\le -CB^2/A^2<0.\]
	Therefore, $\mathbb{E}(X^2)-\mathbb{E}(X)-[\mathbb{E}(X)]^2=e^{-2\lambda}\{H(1)H''(1)-[H'(1)]^2\}<0$, and $\text{VMR}(X)<1$. 
	
	We have proved in (1) that if $h_{n+1}=h_n$ for all $n$, then $\text{VMR}(X)=1$. Thus when $h_{n+1}\le h_n$ for all $n$, $\text{VMR}(X)=1$ if and only if $h_{n+1}=h_n$ for all $n$.
\end{proof}

In sum, for the Markov chain model of one gene (by assuming the expression to be autonomous), when we have the stationary distribution from single-cell non-interventional one-time gene expression data, we can calculate the VMR of $X$. $\text{VMR}(X)>1$ means the existence of positive autoregulation (while negative autoregulation might still be possible at different expression levels), and $\text{VMR}(X)<1$ means the existence of negative autoregulation (while positive autoregulation might still be possible at different expression levels). $\text{VMR}(X)=1$ means either (1) no autoregulation exists; or (2) both positive autoregulation and negative autoregulation exist (at different expression levels). In reality, many genes are non-autonomous, and transcriptional/translational bursting can make the VMR to be larger than $100$ \cite{paulsson2005models}. Since Proposition~\ref{prop2} does not apply to non-autonomous cases, such genes might not have autoregulations.

\section{Scenario of multiple entangled genes}
\label{multi}
\subsection{Setup}
We consider $m$ genes $V_1,\ldots,V_m$ for a single cell. Denote their expression levels by random variables $X_1,\ldots,X_m$. The change of $X_i$ can depend on $X_j$ (mutual regulation) and $X_i$ itself (autoregulation). Since these genes regulate each other, and their expression levels are not fixed, we cannot consider them separately. If the expression of gene $V_k$ is non-autonomous, we also need to add its interior factors (gene state and/or mRNA count) into $X_1,\ldots,X_m$.

We can use a continuous-time one-step Markov chain on $(\mathbb{Z}^*)^m$ to describe the dynamics. Each state of this Markov chain, $(X_1=n_1,\ldots,X_i=n_i,\ldots,X_m=n_m)$, can be abbreviated as $\boldsymbol{n}=(n_1,\ldots,n_i,\ldots,n_m)$. For gene $V_i$, the transition rate of $n_i-1\to n_i$ is $f_i(\boldsymbol{n})$, and the transition rate of $n_i\to n_i-1$ is $g_i(\boldsymbol{n})n_i$. Transitions with more than one step are not allowed. The master equation of this process is 
\begin{equation}
	\label{me}
	\begin{split}
		\frac{\mathrm{d}\mathbb{P}(\boldsymbol{n})}{\mathrm{d}t}=&\sum_i\mathbb{P}(n_1,\ldots,n_i+1,\ldots,n_m)g_i(n_1,\ldots,n_i+1,\ldots,n_m)(n_i+1)\\
		+&\sum_i\mathbb{P}(n_1,\ldots,n_i-1,\ldots,n_m)f_i(\boldsymbol{n})\\
		-&\mathbb{P}(\boldsymbol{n})\sum_i[f_i(n_1,\ldots,n_i+1,\ldots,n_m)+g_i(\boldsymbol{n})n_i].
	\end{split}		
\end{equation}	
Define $\boldsymbol{n}_{\bar{i}}=(n_1,\ldots,n_{i-1},n_{i+1},\ldots,n_m)$. Define $h_i(\boldsymbol{n})=f_i(\boldsymbol{n})/g_i(\boldsymbol{n})$ to be the relative growth rate of gene $V_i$. Autoregulation means for some fixed $\boldsymbol{n}_{\bar{i}}$,  $h_i(\boldsymbol{n})$ is (locally) increasing/decreasing  with $n_i$, thus $f_i(\boldsymbol{n})$ increases/decreases and/or $g_i(\boldsymbol{n})$ decreases/increases with $n_i$. For the non-autonomous scenario, another possibility for autoregulation is that $V_i$ can affect its interior factors (gene state and/or mRNA count).

\subsection{Theoretical results}
With expression data for multiple genes, there are various methods to infer the regulatory relationships between different genes, so that the GRN can be reconstructed \cite{wang2022inference}. In the GRN, if there is a directed path from gene $V_i$ to gene $V_j$, meaning that $V_i$ can directly or indirectly regulate $V_j$, then $V_i$ is an ancestor of $V_j$, and $V_j$ is a descendant of $V_i$.

Fix a gene $V_k$ in a GRN. We consider a simple case that $V_k$ is not contained in any directed cycle (feedback loop), which means no gene is both an ancestor and a descendant of $V_k$, such as PIP2 in Fig.~\ref{grn}. This means $V_k$ itself is a strongly connected component of the GRN. This condition is automatically satisfied if the GRN has no directed cycle. If the expression of $V_k$ is non-autonomous, we need to add the interior factors (gene state and/or mRNA count) of $V_k$ into $V_1,\ldots,V_m$, and it is acceptable that $V_k$ regulates its interior factors. In this case, if the one-step model holds, we can prove that if $V_k$ does not regulate itself, meaning that $h_k(\boldsymbol{n})$ is a constant for fixed $\boldsymbol{n}_{\bar{k}}$ and different $n_k$, and $X_k$ does not affect its interior factors (if non-autonomous), then $\text{VMR}(X_k)\ge 1$. The reason is that $\text{VMR}<1$ requires either a feedback loop or autoregulation. Certainly, $\text{VMR}<1$ might also mean that the one-step model fails. One intuition is to assume the transitions of $V_{\bar{k}}$ are extremely slow, so that $V_k$ is approximately the average of many Poisson variables. It is easy to verify that the average of Poisson variables has $\text{VMR}\ge 1$. We need to assume that the per molecule degradation rate $g_k(\cdot)$ for $V_k$ is not affected by $V_1,\ldots,V_m$, which is not always true in reality \cite{karamyshev2018lost}. With this result, when $\text{VMR}<1$, there might be autoregulation.
\begin{proposition}
	\label{np}
	Consider the one-step Markov chain model for multiple genes, described by Eq.~\ref{me}. Assume the GRN has no directed cycle, or at least there is no directed cycle that contains gene $V_k$. Assume $g_k(\cdot)$ is a constant for all $\boldsymbol{n}$. If $V_k$ has no autoregulation, meaning that $h_k(\cdot)$ and $f_k(\cdot)$ do not depend on $n_k$, and $V_k$ does not regulate its interior factors (gene state and/or mRNA count), then $V_k$ has $\text{VMR}\ge 1$. Therefore, $V_k$ has $\text{VMR}< 1$ means $V_k$ has autoregulation, or the one-step model fails.
\end{proposition}
Paulsson et al. study a similar problem \cite{hilfinger2016constraints,yan2019kinetic}, and they state Proposition~\ref{np} in an unpublished work. Proposition~\ref{np} also appears in a preprint by Mahajan et al. \cite{mahajan2021topological}, but the proof is based on a linear noise approximation, which requires that $f_k(\cdot)$ is linear with $\boldsymbol{n}_{\bar{k}}$. We propose a rigorous proof independently.
\begin{proof}
	Denote the expression level of $V_k$ by $W$. Assume the ancestors of $V_k$ are $V_1,\ldots,V_l$. For simplicity, denote the expression levels of $V_1,\ldots,V_l$ by a (high-dimensional) random variable $Y$. Assume $V_k$ has no autoregulation. Since $V_k$ does not regulate $V_1,\ldots,V_l$, $W$ does not affect $Y$. Denote the transition rate from $Y=i$ to $Y=j$ by $q_{ij}\ge 0$. Stipulate that $q_{ii}=-\sum_{j\ne i}q_{ij}$. When $Y=i$, the transition rate from $W=n$ to $W=n+1$ is $F_i$ (does not depend on $n$), and the transition rate from $W=n$ to $W=n-1$ is $G$. 
	
	The master equation of this process is 
	\begin{equation*}
		\begin{split}
			&\frac{\mathrm{d}\mathbb{P}[W(t)=n,Y(t)=i]}{\mathrm{d}t}\\
			=&\mathbb{P}[W(t)=n-1,Y(t)=i]F_i+\mathbb{P}[W(t)=n+1,Y(t)=i]G(n+1)\\
			&+\sum_{j\ne i}\mathbb{P}[W(t)=n,Y(t)=j]q_{ji}-\mathbb{P}[W(t)=n,Y(t)=i](F_i+Gn+\sum_{j\ne i}q_{ij}).
		\end{split}
	\end{equation*}
	Assume there is a unique stationary probability distribution $P_{n,i}=\lim_{t\to \infty}\mathbb{P}[W(t)=n,Y(t)=i]$. This can be guaranteed by assuming the process to be irreducible. Then we have 
	\begin{equation}
		\label{eq1}
		P_{n,i}\Big[F_i+Gn+\sum_{j}q_{ij}\Big]=P_{n-1,i}F_i+P_{n+1,i}G(n+1)+\sum_{j}P_{n,j}q_{ji}.
	\end{equation}
	Define $P_i=\sum_n P_{n,i}$. Sum over $n$ for Eq.~\ref{eq1} to obtain 
	\begin{equation}
		\label{eqz1}
		P_i\sum_{j}q_{ij}=\sum_{j}P_jq_{ji},
	\end{equation}
	meaning that $P_i$ is the stationary probability distribution of $Y$.
	
	Define $W_i$ to be $W$ conditioned on $Y=i$ at stationarity. Then $\mathbb{P}(W_i=n)=\mathbb{P}(W=n\mid Y=i)=P_{n,i}/P_i$, and $\mathbb{E}(W_i)=\sum_n n P_{n,i}/P_i$. Multiply Eq.~\ref{eq1} by $n$ and sum over $n$ to obtain 
	\begin{equation}
		\label{eq2}
		\Big(G+\sum_j q_{ij}\Big)P_i\mathbb{E}(W_i)=F_iP_i+\sum_j q_{ji}P_j\mathbb{E}(W_j).
	\end{equation}
Here and in the following, we repeatedly apply the tricks of splitting $n$ and shifting the index of summation. For example, 
\begin{equation*}
	\begin{split}
		&\sum_{n=1}^{\infty} P_{n-1,i}F_i n-\sum_{n=1}^{\infty} P_{n,i}F_i n\\
		=&\sum_{n=1}^{\infty} P_{n-1,i}F_i (n-1)+\sum_{n=1}^{\infty} P_{n-1,i}F_i-\sum_{n=1}^{\infty} P_{n,i}F_i n\\
		=&\sum_{n-1=0}^{\infty} P_{n-1,i}F_i (n-1)+\sum_{n-1=0}^{\infty} P_{n-1,i}F_i-\sum_{n=0}^{\infty} P_{n,i}F_i n\\
		=&\sum_{n=0}^{\infty} P_{n,i}F_i n+F_i\sum_{n=0}^{\infty} P_{n,i}-\sum_{n=0}^{\infty} P_{n,i}F_i n=F_iP_i.
	\end{split}
\end{equation*}

	Sum over $i$ for Eq.~\ref{eq2} to obtain
	\begin{equation}
		\label{eq3}
		G\sum_iP_i\mathbb{E}(W_i)=\sum_iF_iP_i.
	\end{equation}
	Multiply Eq.~\ref{eq1} by $n^2$ and sum over $n$ to obtain 
	\begin{equation}
		\label{eq4}
		\Big(2G+\sum_j q_{ij}\Big)P_i\mathbb{E}(W_i^2)=F_iP_i+(2F_i+G)P_i\mathbb{E}(W_i)+\sum_j q_{ji}P_j\mathbb{E}(W_j^2).
	\end{equation}
	Sum over $i$ for Eq.~\ref{eq4} to obtain
	\begin{equation}
		\label{eq5}
		2G\sum_iP_i\mathbb{E}(W_i^2)=\sum_iF_iP_i+2\sum_i F_iP_i\mathbb{E}(W_i)+G\sum_iP_i\mathbb{E}(W_i).
	\end{equation}
	Multiply Eq.~\ref{eq2} by $\mathbb{E}(W_i)$ and sum over $i$ to obtain
	\begin{equation}
		\label{eqz2}
		\begin{split}
			&G\sum_iP_i[\mathbb{E}(W_i)]^2+\sum_{i,j}P_iq_{ij}[\mathbb{E}(W_i)]^2\\
			=&\sum_iF_iP_i\mathbb{E}(W_i)+\sum_{i,j}P_jq_{ji}\mathbb{E}(W_i)\mathbb{E}(W_j).
		\end{split}
	\end{equation}
	Then we have
	\begin{equation} 
		\label{eq6}
		\begin{split}
			&\sum_iF_iP_i\mathbb{E}(W_i)-G\sum_iP_i[\mathbb{E}(W_i)]^2 \\
			=&\sum_{i,j}P_iq_{ij}[\mathbb{E}(W_i)]^2-\sum_{i,j}P_jq_{ji}\mathbb{E}(W_i)\mathbb{E}(W_j)\\
			=&\frac{1}{2}\Big\{\sum_{i,j}P_iq_{ij}[\mathbb{E}(W_i)]^2+\sum_i[\mathbb{E}(W_i)]^2\sum_jP_iq_{ij}-2\sum_{i,j}P_iq_{ij}\mathbb{E}(W_i)\mathbb{E}(W_j)\Big\}\\
			=&\frac{1}{2}\Big\{\sum_{i,j}P_iq_{ij}[\mathbb{E}(W_i)]^2+\sum_i[\mathbb{E}(W_i)]^2\sum_jP_jq_{ji}-2\sum_{i,j}P_iq_{ij}\mathbb{E}(W_i)\mathbb{E}(W_j)\Big\}\\
			=&\frac{1}{2}\Big\{\sum_{i,j}P_iq_{ij}[\mathbb{E}(W_i)]^2+\sum_{i,j}P_iq_{ij}[\mathbb{E}(W_j)]^2-2\sum_{i,j}P_iq_{ij}\mathbb{E}(W_i)\mathbb{E}(W_j)\Big\}\\
			=&\frac{1}{2}\sum_{i,j}P_iq_{ij}[\mathbb{E}(W_i)-\mathbb{E}(W_j)]^2\ge 0.
		\end{split}
	\end{equation}
	Here the first equality is from Eq.~\ref{eqz2}, the third equality is from Eq.~\ref{eqz1}, and other equalities are equivalent transformations.
	
	Now we have  
	\begin{equation}
		\begin{split}
			&\mathbb{E}(W^2)-\mathbb{E}(W)-[\mathbb{E}(W)]^2\\
			=&\sum_iP_i\mathbb{E}(W_i^2)-\sum_iP_i\mathbb{E}(W_i)-\Big[\sum_iP_i\mathbb{E}(W_i)\Big]^2\\
			=&\frac{1}{G}\sum_i F_iP_i\mathbb{E}(W_i)+\sum_i P_i\mathbb{E}(W_i)-\sum_iP_i\mathbb{E}(W_i)-\Big[\sum_iP_i\mathbb{E}(W_i)\Big]^2\\
			\ge &\sum_iP_i[\mathbb{E}(W_i)]^2-\Big[\sum_iP_i\mathbb{E}(W_i)\Big]^2\\
			=&\Big(\sum_iP_i\Big)\sum_iP_i[\mathbb{E}(W_i)]^2-\Big[\sum_iP_i\mathbb{E}(W_i)\Big]^2\ge 0,\\
		\end{split}
	\label{fe}
	\end{equation}
	where the first equality is by definition, the second equality is from Eqs.~\ref{eq3},\ref{eq5}, the first inequality is from Eq.~\ref{eq6}, the third equality is from $\sum_i P_i=1$, and the second inequality is the Cauchy-Schwarz inequality.
	
	Since $\mathbb{E}(W^2)-[\mathbb{E}(W)]^2\ge \mathbb{E}(W)$, $\text{VMR}(W)=\{\mathbb{E}(W^2)-[\mathbb{E}(W)]^2\}/ \mathbb{E}(W)\ge 1$.
\end{proof}

\begin{remark}
	In gene expression, the total noise ($\sigma^2(X)/(\mathbb{E}X)^2$) can be decomposed into intrinsic (cellular) noise and extrinsic (environmental) noise \cite{baudrimont2019contribution,thomas2019intrinsic,ham2020extrinsic,lin2021disentangling,wang2019roles}. Inspired by that, we can decompose the VMR into intrinsic and extrinsic components. Denote intrinsic and extrinsic stochastic factors as $I,E$, and the expression level $X$ is a deterministic function of these factors: $X=X(I,E)$. Then 
	\[\text{VMR}_{\text{int}}=\frac{\mathbb{E}_E(\mathbb{E}_{I\mid E}X^2)-
		\mathbb{E}_E(\mathbb{E}_{I\mid E}X)^2}{\mathbb{E}X},\]
	\[\text{VMR}_{\text{ext}}=\frac{\mathbb{E}_E(\mathbb{E}_{I\mid E}X)^2-[\mathbb{E}_E(\mathbb{E}_{I\mid E}X)]^2}{\mathbb{E}X},\]
	where $\mathbb{E}_{I\mid E}$ is the expectation conditioned on $E$. This decomposition might lead to further understanding of Proposition~\ref{np}.
\end{remark}

We hypothesize that the requirement for $g_k(\cdot)$ in Proposition~\ref{np} can be dropped:
\begin{conjecture}
	\label{conj2}
	Assume $V_k$ is not contained in a directed cycle in the GRN, and $V_k$ does not regulate its interior factors (gene state and/or mRNA count). If $V_k$ has no autoregulation, meaning that $h_k(\cdot)$ does not depend on $n_k$ (but might depend on $\boldsymbol{n}_{\bar{k}}$), then $V_k$ has $\text{VMR}\ge 1$. 
\end{conjecture}

The main obstacle of proving this conjecture is that the second equality in Eq.~\ref{fe} does not hold. The reason is that $G_i$ cannot be extracted from the summation, and we cannot link $\sum_iP_i\mathbb{E}(W_i^2)$ and $\sum_i G_iP_i\mathbb{E}(W_i^2)$.

If the GRN has directed cycles, there is a result by Paulsson et al. \cite{hilfinger2016constraints,yan2019kinetic}, which is proved under first-order approximations of covariances. The general case (when the approximations do not apply) has been numerically verified but not proved yet:
\begin{conjecture}
	Assume for each $V_i$, $g_i(\cdot)$ does not depend on $\boldsymbol{n}$, and $f_i(\cdot)$ does not depend on $n_i$ (no autoregulation). Then for at least one gene $V_j$, we have $\text{VMR}\ge 1$ \cite{hilfinger2016constraints,yan2019kinetic}.
	\label{conj3}
\end{conjecture}
Due to the existence of directed cycles, one gene can affect itself through other genes, and we cannot study them separately.

Notice that Conjecture~\ref{conj3} does not hold if $g_i$ depends on $\boldsymbol{n}_{\bar{i}}$:
\begin{example}
	Consider a one-step Markov chain that satisfies Eq.~\ref{me}, where $m=2$, $f_1(n_2)=g_1(n_2)=1$ for $n_2=2$, $f_1(n_2)=g_1(n_2)=0$ for $n_2\ne 2$, and $f_2(n_1)=g_2(n_1)=1$ for $n_1=2$, $f_2(n_1)=g_2(n_1)=0$ for $n_1\ne 2$. The initial state is $(n_1=2,n_2=2)$. Then $\text{VMR}=2e/(4e-1)\approx 0.55$ for both genes (see Appendix~\ref{a0.2} for details).
	\label{ex2}
\end{example}

Assume Conjecture~\ref{conj3} is correct. For $m$ genes, if we find that VMR for each gene is less than $1$, then we can infer that autoregulation exists, although we do not know which gene has autoregulation. Another possibility is that the one-step model fails.

\section{Applying theoretical results to experimental data}
\label{app}

\begin{algorithm}[!htbp]
	\caption{Detailed workflow of inferring autoregulation with gene expression data.}
	\label{alg}
	\vspace{-\bigskipamount}
	\ \\
	\begin{enumerate}
		\item \textbf{Input} 
		
		\quad Single-cell non-interventional one-time expression data for genes $V_1,\ldots,V_m$
		
		\quad The structure of the GRN that contains $V_1,\ldots,V_m$
		\item \textbf{Calculate} the VMR of each $V_k$\
		
		\item \textbf{If} $V_k$ is not in a directed cycle (like PIP2 in Fig.~\ref{grn}) and $\text{VMR}<1$
		
		\quad   \textbf{Output} $V_k$ has autoregulation
		
		\quad   // Assume the degradation of $V_k$ is not regulated by $V_1,\ldots,V_m$
		
		\textbf{Else} 
		
		\quad   \textbf{If} $V_k$ has no ancestor in the GRN (like PIP3 in Fig.~\ref{grn}) and $\text{VMR}>1$
		
		\quad\quad   \textbf{Output} $V_k$ has autoregulation
		
		\quad\quad //Assume the expression of $V_k$ is autonomous
		
		\quad   \textbf{Else}
		
		\quad\quad   \textbf{Output} We cannot determine whether $V_k$ has autoregulation
		
		\quad   \textbf{End} of if
		
		\textbf{End} of if

	\end{enumerate}
\end{algorithm}

We summarize our theoretical results into Algorithm~\ref{alg}. Proposition~\ref{prop2} applies to a gene that has no ancestor in the GRN. However, it requires the corresponding gene has autonomous expression (or the transition rates of gene states are high enough, so that the non-autonomous process is close to an autonomous process), which is difficult to validate and often does not hold in reality. Thus the inference result by Proposition~\ref{prop2} for $\text{VMR}>1$ (positive autoregulation) is not very reliable. When $\text{VMR}<1$ and Proposition~\ref{prop2} could apply, we should instead apply Proposition~\ref{np} to determine the existence of autoregulation, since Proposition~\ref{np} does not require the expression to be autonomous, thus being much more reliable, although it may fail if the one-step model does not hold. Proposition~\ref{np} applies when the gene is not in a feedback loop and has $\text{VMR}<1$. Notice that our result cannot determine that a gene has no autoregulation. 

For a given gene without autoregulation, its expression level satisfies a Poisson distribution, and VMR is $1$. If we have $n$ samples of its expression level, then the sample VMR (sample variance divided by sample mean) asymptotically satisfies a Gamma distribution $\Gamma[(n-1)/2,2/(n-1)]$, and we can determine the confidence interval of sample VMR \cite{eden2010drawing}. If the sample VMR is out of this confidence interval, then we know that VMR is significantly different from $1$, and Propositions~\ref{prop2},\ref{np} might apply. 

We apply our method to four groups of single-cell non-interventional one-time gene expression data from experiments, where the corresponding GRNs are known. Notice that we need to convert indirect measurements into protein/mRNA count. See Table~\ref{tab} for our inference results and theoretical/experimental evidence that partially validates our results. See Appendix~\ref{app1} for details. There are 186 genes in these four data sets, and we can only determine that 12 genes have autoregulation (7 genes determined by Proposition~\ref{prop2}, and 5 genes determined by Proposition~\ref{np}). Not every VMR is less than $1$, so that Conjecture~\ref{conj3} does not apply. For the other 174 genes, (1) some of them are not contained in the known GRN, and we cannot determine if they are in directed cycles; (2) some of them are in directed cycles; (3) some of them have ancestors, and we cannot reject the hypothesis that $\text{VMR}\ge 1$; (4) some of them have no ancestors, and we cannot reject the hypothesis that $\text{VMR}= 1$. Therefore, Proposition~\ref{prop2} and Proposition~\ref{np} do not apply, and we do not know whether they have autoregulation. 

In some cases, we have experimental evidence that some genes have autoregulation, so that we can partially validate our inference results. Nevertheless, as discussed in the Introduction, there is no gold standard to evaluate our inference results. Besides, Proposition~\ref{np} requires that the one-step model holds, which we cannot verify.

In the data set by Guo et al. \cite{guo2010resolution}, Sanchez-Castillo et al. \cite{sanchez2018bayesian} inferred that 17 of 39 genes have autoregulation, and 22 genes do not have autoregulation. We infer that 5 genes have autoregulation, and 34 genes cannot be determined. Here 3 genes are shared by both inference results to have autoregulation. Consider a random classifier that randomly picks 5 genes and claims they have autoregulation. Using Sanchez-Castillo et al. as the standard, this random classifier has probability $62.55\%$ to be worse than our result, and $10.17\%$ to be better than our result. Thus our inference result is better than a random classifier, but the advantage is not substantial.

\begin{table}[]	
	\begin{tabular}{lllll}
		Source          &\begin{tabular}[c]{@{}l@{}}Propo-\\ sition~\ref{prop2}\end{tabular}                                                     & \begin{tabular}[c]{@{}l@{}}Propo-\\ sition~\ref{np}\end{tabular}                                                              & Theory                                                                 & Experiment                                                     \\
		\hline
		\begin{tabular}[c]{@{}l@{}}Guo\\ et al. \cite{guo2010resolution} \end{tabular}
		& \begin{tabular}[c]{@{}l@{}}FN1\\ \textbf{HNF4A}\end{tabular}               & \begin{tabular}[c]{@{}l@{}} \textbf{TCFAP2C}\\ \textbf{BMP4}\\ CREB312\end{tabular} & \begin{tabular}[c]{@{}l@{}}BMP4 \cite{sanchez2018bayesian}\\ HNF4A \cite{sanchez2018bayesian}\\  TCFAP2C \cite{sanchez2018bayesian}\end{tabular} & \begin{tabular}[c]{@{}l@{}}BMP4 \cite{pramono2016thrombopoietin}\\ HNF4A \cite{chahar2014chromatin}\\ TCFAP2C \cite{kidder2010examination}\end{tabular} \\
		\hline
		\begin{tabular}[c]{@{}l@{}}Psaila\\ et al. \cite{psaila2016single}\end{tabular}
		& \begin{tabular}[c]{@{}l@{}}BIM\\ CCND1\\ \textbf{ECT2}\\ PFKP\end{tabular} &                                                                     &                                                                        & ECT2 \cite{hara2006cytokinesis}                                                           \\
		\hline
		\begin{tabular}[c]{@{}l@{}}Moignard\\ et al. \cite{moignard2015decoding}\end{tabular}
		&                                                                   & \begin{tabular}[c]{@{}l@{}}EIF2B1\\ HOXD8\end{tabular}          &                                                                        &                                                                \\
		\hline
		\begin{tabular}[c]{@{}l@{}}Sachs\\ et al. \cite{sachs2005causal}\end{tabular}
		& PIP3                                                              &                                                                          &                                                                        &                                                               
	\end{tabular}
	\caption{The autoregulation inference results by our method on four data sets. Source column is the paper that contains this data set. Proposition~\ref{prop2} column is the genes that can be only inferred by Proposition~\ref{prop2} to have autoregulation. Proposition~\ref{np} column is the genes that can be inferred by Proposition~\ref{np} to have autoregulation. Theory column is the genes inferred by both our method and other theoretical works to have autoregulation. Experiment column is the genes inferred by both our method and other experimental works to have autoregulation. \textbf{Bold} font means the inferred gene with autoregulation is validated by other results. Details can be found in Appendix~\ref{app1}.}
	\label{tab}
\end{table}



\section{Conclusions}
\label{con}
For a single gene that is not affected by other genes, or a group of genes that form a connected GRN, we develop rigorous theoretical results (without applying approximations) to determine the existence of autoregulation. These results generalize known relationships between autoregulation and VMR by dropping restrictions on parameters. Our results only depend on VMR, which is easy to compute and more robust than other complicated statistics. We also apply our method to experimental data and detect some genes that might have autoregulation. 

Our method requires \textbf{independent} and \textbf{identically} distributed samples from the \textbf{exact} \textbf{stationary} distribution of a \textbf{fully observed} Markov chain, plus a known \textbf{GRN}. Proposition~\ref{prop2} requires that the expression is \textbf{autonomous}. Proposition~\ref{np} requires that the Markov chain model is  \textbf{one-step}, the GRN has \textbf{no} directed \textbf{cycle}, and \textbf{degradation} is \textbf{not regulated}. If our inference fails, then some requirements are not met: (1) cells might affect each other, making the samples dependent; (2) cells are heterogeneous; (3) the measurements have extra errors; (4) the cells are not at stationarity; (5) there exist unobserved variables that affect gene expression; (6) the GRN is inferred by a theoretical method, which can be interfered by the existence of autoregulation; (7) the expression is non-autonomous; (8) the Markov chain is multi-step; (9) the GRN has unknown directed cycles; (10) the degradation rate is regulated by other genes. Such situations, especially the unobserved variables, are unavoidable. Therefore, current data might not satisfy these requirements, and our inference results should be interpreted as informative findings, not ground truths. 

There are some known methods that overcome the above obstacles, and there are also some possible solutions that might appear in the future. (1) The dependency can be solved by better measurements for isolated cells that do not affect each other. In fact, the relationship between autoregulation and cell-cell interaction has been studied \cite{levenberg1998long}. (2) About cell heterogeneity, we prove a result in Appendix~\ref{app2} that if several cell types have $\text{VMR}\ge 1$, then for a mixed population of such cell types, we still have $\text{VMR}\ge1$. Therefore, cell heterogeneity does not fail Proposition~\ref{np}, since $\text{VMR}<1$ for the mixture of several cell types means $\text{VMR}<1$ for at least one cell type. (3) With the development of experimental technologies, we expect that the measurement error can decrease. (4) Some works study autoregulation in non-stationary situations \cite{cao2020analytical,swain2002intrinsic,skinner2016single,jia2021frequency}. (5) Since hidden variables hurt any mechanism-based models, we can develop methods (especially with machine learning tools) that determine autoregulation based on similarities between gene expression profiles \cite{wang2021inference,yang2020machine,wang2022two,wang2021measuring,wang2023online}. (6) Some GRN inference methods can also determine the existence of autoregulation \cite{sanchez2018bayesian}. (7) Many methods (including our Proposition~\ref{np}) work in non-autonomous situations. (8) Some works study multi-step models \cite{braichenko2021distinguishing,karmakar2021effect,voliotis2008fluctuations}. (9) We expect the appearance of more advanced GRN inference methods. (10) If probabilists can prove Conjecture~\ref{conj2}, then the restriction on degradation rate can be lifted.

In fact, other theoretical works that determine gene autoregulation, or general gene regulation, also need various assumptions and might fail. Nevertheless, with the development of experimental technologies and theoretical results, we believe that some obstacles will be lifted, and our method will be more applicable in the future. Besides, our method can be further developed and combined with other methods.

\section*{Acknowledgments}

Y.W. would like to thank Jiawei Yan for fruitful discussions, and Xiangting Li, Zikun Wang, Mingtao Xia for helpful comments. The authors would like to thank some anonymous reviewers for their wise suggestions.

\section*{Declarations}

\begin{itemize}
\item Funding: This research was partially supported by NIH grant R01HL146552 (Y.W.). 
\item Competing interests: The authors declare no conflict of interest.

\item Data and code availability: All code files are available in

https://github.com/YueWangMathbio/Autoregulation.

\end{itemize}

\appendix
	
\section{Simulation results}
\label{sim}
\subsection{Test Proposition~\ref{prop2} without autoregulation}
\textbf{Simulation 1}: Consider the Markov chain in Example~\ref{ex1} with $k=1,b=0,c=1$. The stationary distribution is Poissonian with parameter $1$. This process has no autoregulation with the true $\text{VMR}=1$. For sample sizes $n=100$, $n=1000$, $n=10000$, we repeat the experiment for $10000$ times and calculate the rate that the sample VMR falls in the $95\%$ confidence interval. See Table~\ref{ts1} for results. Proposition~\ref{prop2} has about $95\%$ probability to produce the correct result that $\text{VMR}=1$ (no autoregulation), since the confidence interval is $95\%$.

\begin{table}[]
	\begin{tabular}{l|lll}
		& $\text{VMR}<1$ & $\text{VMR}=1$ (true) & $\text{VMR}>1$ \\ \hline
		$n=100$   & $2.5\%$               & $95.3\%$               & $2.2\%$                  \\
		$n=1000$  & $2.3\%$               & $95.2\%$              & $2.5\%$                  \\
		$n=10000$ & $2.5\%$               & $95.1\%$               & $2.4\%$                 
	\end{tabular}
\caption{Success rates for determining VMR in Simulation 1. For different sample sizes $n$, calculate the rate that VMR is different from $1$.}
\label{ts1}
\end{table}

\subsection{Test Proposition~\ref{prop2} with autoregulation}
\textbf{Simulation 2}: Consider the Markov chain in Example~\ref{ex1} with $k=1,b=1,c=2$. The stationary distribution is geometric with parameter $0.5$. This process has positive autoregulation with the true $\text{VMR}=2$. For sample sizes $n=100$, $n=1000$, $n=10000$, we repeat the experiment for $10000$ times and calculate the rate that the sample VMR falls in the $95\%$ confidence interval. See Table~\ref{ts2} for results. When $n$ is not too small, Proposition~\ref{prop2} always produces the correct result that $\text{VMR}>1$ (positive autoregulation).

\begin{table}[]
	\begin{tabular}{l|lll}
		& $\text{VMR}<1$ & $\text{VMR}=1$ & $\text{VMR}>1$ (true) \\ \hline
		$n=100$   & $0\%$               & $1.9\%$               & $98.1\%$                  \\
		$n=1000$  & $0\%$               & $0\%$              & $100\%$                  \\
		$n=10000$ & $0\%$               & $0\%$               & $100\%$                 
	\end{tabular}
	\caption{Success rates for determining VMR in Simulation 2. For different sample sizes $n$, calculate the rate that VMR is different from $1$.}
	\label{ts2}
\end{table}

\subsection{Test Proposition~\ref{np} without autoregulation}
\textbf{Simulation 3}: Consider a Markov chain $(G,M)$ that satisfies Eq.~\ref{me}. $G$ can take $0$ and $1$, and $M$ can take values in $\mathbb{Z}$. $G$ does not depend on $M$, and transition rates are both $10^{-10}$ for $G=0 \to G=1$ and $G=1\to G=0$. Restricted on $G=0$, $M$ is the same as Example~\ref{ex1} with $k=1,b=0,c=1$. Restricted on $G=1$, $M$ is the same as Example~\ref{ex1} with $k=2,b=0,c=1$. The stationary distribution is the average of two Poisson distributions with parameters $1$ and $2$. This process has no autoregulation with the true $\text{VMR}=1.167$. For sample sizes $n=100$, $n=1000$, $n=10000$, we repeat the experiment for $10000$ times and calculate the rate that the sample VMR falls in the $95\%$ confidence interval. See Table~\ref{ts3} for results. When $n$ increases, we are very likely to obtain the correct $\text{VMR}>1$, but Proposition~\ref{np} cannot determine whether autoregulation exists.

\begin{table}[]
	\begin{tabular}{l|lll}
		& $\text{VMR}<1$ & $\text{VMR}=1$ & $\text{VMR}>1$ (true) \\ \hline
		$n=100$   & $0.2\%$               & $81.6\%$               & $18.2\%$                  \\
		$n=1000$  & $0\%$               & $7.7\%$              & $92.3\%$                  \\
		$n=10000$ & $0\%$               & $0\%$               & $100\%$                 
	\end{tabular}
	\caption{Success rates for determining VMR in Simulation 3. For different sample sizes $n$, calculate the rate that VMR is different from $1$.}
	\label{ts3}
\end{table}

\subsection{Test Proposition~\ref{np} with autoregulation}
\textbf{Simulation 4}: Consider a Markov chain $(G,M)$ that satisfies Eq.~\ref{me}. $G$ can take $0$ and $1$, and $M$ can take values in $\mathbb{Z}$. $G$ does not depend on $M$, and transition rates are both $10^{-10}$ for $G=0 \to G=1$ and $G=1\to G=0$. Restricted on $G=0$, $M$ is the same as Example~\ref{ex1} with $k=10,b=-1,c=1$. Restricted on $G=1$, $M$ is the same as Example~\ref{ex1} with $k=10,b=-1,c=2$. This process has autoregulation with the true $\text{VMR}=0.733$. For sample sizes $n=100$, $n=1000$, $n=10000$, we repeat the experiment for $10000$ times and calculate the rate that the sample VMR falls in the $95\%$ confidence interval. See Table~\ref{ts4} for results. When $n$ is not too small, Proposition~\ref{np} always produces the correct result that $\text{VMR}<1$ (autoregulation).

\begin{table}[]
	\begin{tabular}{l|lll}
		& $\text{VMR}<1$ (true) & $\text{VMR}=1$ & $\text{VMR}>1$ \\ \hline
		$n=100$   & $56.6\%$               & $43.4\%$               & $0\%$                  \\
		$n=1000$  & $100\%$               & $0\%$              & $0\%$                  \\
		$n=10000$ & $100\%$               & $0\%$              & $0\%$                 
	\end{tabular}
	\caption{Success rates for determining VMR in Simulation 4. For different sample sizes $n$, calculate the rate that VMR is different from $1$.}
	\label{ts4}
\end{table}
	
\section{Details of examples}
\label{app0}
\subsection{Details of Example~\ref{ex1}}
\label{a0.1}
In Example~\ref{ex1}, the stationary distribution $P_n$ exists, and satisfies
\begin{equation}
	[b(n-1)+k]P_{n-1}=cnP_n.
	\label{eqa}
\end{equation}

Taking summation for Eq.~(\ref{eqa}), we have
\begin{equation*}
	\begin{split}
		&\sum_{n=1}^{\infty}[b(n-1)+k]P_{n-1}=\sum_{n=1}^{\infty}cnP_n,\\
		\Rightarrow & b\sum_{n=0}^{\infty}nP_n+k\sum_{n=0}^{\infty}P_n=c\sum_{n=0}^{\infty}nP_n,\\
		\Rightarrow & b\mathbb{E}X+k=c\mathbb{E}X.
	\end{split}
\end{equation*}
Thus $\mathbb{E}X=k/(c-b)$.

Also, multiplying Eq.~(\ref{eqa}) by $n-1$ and taking summation, we have
\begin{equation*}
	\begin{split}
		&\sum_{n=1}^{\infty}[b(n-1)^2+k(n-1)]P_{n-1}=\sum_{n=1}^{\infty}cn^2P_n-\sum_{n=1}^{\infty}cnP_n,\\
		\Rightarrow & b\sum_{n=0}^{\infty}n^2P_n+k\sum_{n=0}^{\infty}nP_n=c\sum_{n=0}^{\infty}n^2P_n-c\sum_{n=0}^{\infty}nP_n,\\
		\Rightarrow & b\mathbb{E}(X^2)+k\mathbb{E}X=c\mathbb{E}(X^2)-c\mathbb{E}X.
	\end{split}
\end{equation*}
Thus $\mathbb{E}(X^2)=(c+k)k/(c-b)^2$, and the variance $\sigma^2(X)=ck/(c-b)^2$. Then $\text{VMR}(X)=1+b/(c-b)$.

\subsection{Details of Example~\ref{ex2}}
\label{a0.2}
In Example~\ref{ex2}, the process is restricted to two lines: $n_1=2$ and $n_2=2$. Since this process has no cycle, it is detailed balanced \cite{wang2020mathematical}, and the stationary distribution $P(n_1,n_2)$ satisfies \[P(k,2)=(k+1)P(k+1,2)\] 
and 
\[P(2,k)=(k+1)P(2,k+1).\]
Restricted on $n_1=2$ or $n_2=2$, the stationary distribution is Poissonian, $P(k,2)=c/(ek!)$.
After normalization, we find that $c=2e/(4e-1)$. Thus
\[P(k,2)=P(2,k)=\frac{2}{k!(4e-1)}.\]
Besides, 
\[P(2,\cdot)=\sum_{k=0}^\infty P(2,k)=\sum_{k=0}^\infty \frac{2}{k!(4e-1)}=\frac{2e}{4e-1}.\]
Then for the first gene $X$, we have
\begin{equation*}
	\begin{split}
		\mathbb{E}X & =\sum_{k\ne 2} kP(k,2)+2P(2,\cdot)=\sum_{k=0}^\infty kP(k,2)+2[P(2,\cdot)-P(2,2)] \\
		 & = \sum_{k=0}^\infty \frac{2k}{k!(4e-1)}+\frac{4e-2}{4e-1}=\frac{6e-2}{4e-1},
	\end{split}
\end{equation*}
and
\begin{equation*}
	\begin{split}
		\mathbb{E}(X^2) & =\sum_{k\ne 2} k^2P(k,2)+4P(2,\cdot)=\sum_{k=0}^\infty k^2P(k,2)+4[P(2,\cdot)-P(2,2)] \\
		& = \sum_{k=0}^\infty \frac{2k(k-1)}{k!(4e-1)}+\sum_{k=0}^\infty \frac{2k}{k!(4e-1)}+\frac{8e-4}{4e-1}=\frac{12e-4}{4e-1}.
	\end{split}
\end{equation*}
Thus 
\[\text{VMR}(X)=\frac{\mathbb{E}(X^2)}{\mathbb{E}X}-\mathbb{E}X=2-\frac{6e-2}{4e-1}=\frac{2}{4e-1}\approx 0.55.\]
Due to symmetry, the other gene also has $\text{VMR}(Y)\approx 0.55$.

\section{Details of applications on experimental data}\label{app1}

In experiments, the expression levels of genes are not directly measured as mRNA or protein counts. Rather, they are measured as cycle threshold (Ct) values or fluorescence intensity values. Such indirect measurements need to be converted. Related details can be found in other papers \cite{jia2017stochastic}.

Guo et al. \cite{guo2010resolution} measured the expression (mRNA) levels of 48 genes for mouse embryo cells at different developmental stages. We consider three groups (16-cell stage, 32-cell stage, 64-cell stage) that have more than 50 samples. Sanchez-Castillo et al. \cite{sanchez2018bayesian} used such data to infer the GRN structure, including autoregulation, but the inferred GRN only contains 39 genes. We cannot guarantee that the other 9 genes have no ancestors in the true GRN (to apply Proposition~\ref{prop2}) or these genes are not contained in directed cycles (to apply Proposition~\ref{np}). Thus we ignore those 9 genes not in this GRN. In the inferred GRN, genes BMP4, CREB312, and TCFAP2C are not contained in directed cycles. In the 16-cell stage group with 75 samples, if there is no autoregulation, then the $95\%$ confidence interval of VMR is $[0.7041,1.3470]$. BMP4 ($\text{VMR}=0.2139$), CREB312 ($\text{VMR}=0.1971$), and TCFAP2C ($\text{VMR}=0.3468$) have significantly small VMR, and we can apply Proposition~\ref{np} to infer that BMP4, CREB312, and TCFAP2C might have autoregulation. In the other two groups, these genes do not have $\text{VMR}<1$, and the results are relatively weak. Besides, in the inferred GRN, genes FN1 and HNF4A have no ancestors. For the 16-cell stage with 75 samples, the VMR of FN1 and HNF4A are $3.4522$ and $1.3599$, outside of the $95\%$ confidence interval $[0.7041,1.3470]$; for the 32-cell stage with 113 samples, the VMR of FN1 and HNF4A are $93.1070$ and $46.7688$, outside of the $95\%$ confidence interval $[0.7554,1.2784]$; for the 64-cell stage with 159 samples, the VMR of FN1 and HNF4A are $117.3059$ and $93.9589$, outside of the $95\%$ confidence interval $[0.7917,1.2322]$. Thus we can apply Proposition~\ref{prop2} to infer that FN1 and HNF4A ($\text{VMR}>1$ for all three cell groups) might have positive autoregulation. Nevertheless, it is more likely that the expressions of FN1 and HNF4A are non-autonomous, and there is no autoregulation. Sanchez-Castillo et al. \cite{sanchez2018bayesian} inferred that BMP4, HNF4A, TCFAP2C have autoregulation. Besides, there is experimental evidence that BMP4 \cite{pramono2016thrombopoietin}, HNF4A \cite{chahar2014chromatin}, TCFAP2C \cite{kidder2010examination} have autoregulation. Therefore, our inference results are partially validated.

Psaila et al. \cite{psaila2016single} measured the expression (mRNA) levels of 90 genes for human megakaryocyte-erythroid progenitor cells. Chan et al. \cite{chan2017gene} inferred the GRN structure (autoregulation not included). In the inferred GRN, genes BIM, CCND1, ECT2, PFKP have no ancestors. BIM has 214 effective samples, and VMR is $187.7$, outside of the $95\%$ confidence interval $[0.8191,1.1987]$. CCND1 has 68 effective samples, and VMR is $111.3$, outside of the $95\%$ confidence interval $[0.6905,1.3660]$. ECT2 has 56 effective samples, and VMR is $8.2$, outside of the $95\%$ confidence interval $[0.6618,1.4069]$. PFKP has 134 effective samples, and VMR is $82.1$, outside of the $95\%$ confidence interval $[0.7742,1.2543]$. Thus we can apply Proposition~\ref{prop2} to infer that BIM, CCND1, ECT2, PFKP might have positive autoregulation. Nevertheless, it is more likely that the expressions of these four genes are non-autonomous, and there is no autoregulation. There is experimental evidence that ECT2 has autoregulation \cite{hara2006cytokinesis}, which partially validates our inference results. No other gene fits the requirement of Proposition~\ref{np}.

Moignard et al. \cite{moignard2015decoding} measured the expression (mRNA) levels of 46 genes for mouse embryo cells. Chan et al. \cite{chan2017gene} inferred the GRN structure (autoregulation not included). Gene EIF2B1 has 3934 effective samples, and VMR is $0.66$, outside of the $95\%$ confidence interval $[0.9563,1.0447]$. Gene EIF2B1 has 12 effective samples, and VMR is $0.24$, outside of the $95\%$ confidence interval $[0.3469,1.9927]$. We can apply Proposition~\ref{np} to infer that EIF2B1 and HOXD8 might have autoregulation. No other gene fits the requirement of Proposition~\ref{prop2}.

Sachs et al. \cite{sachs2005causal} measured the expression (protein) levels of 11 genes in the RAF signaling pathway for human T cells. The measurements were repeated for 14 groups of cells under different interventions. Werhli et al. \cite{werhli2006comparative} inferred the GRN structure (autoregulation not included). In the inferred GRN (Fig.~\ref{grn}), PIP3 gene has no ancestor, and its VMRs in all 14 groups are larger than $5$, while the $95\%$ confidence intervals for all 14 groups are contained in $[0.8,1.2]$. Therefore, we can apply Proposition~\ref{prop2} and infer that PIP3 might have positive autoregulation. Nevertheless, it is more likely that the expression of PIP3 is non-autonomous, and there is no autoregulation. No other gene fits the requirement of Proposition~\ref{np}.

\section{Heterogeneity and VMR}
\label{app2}
\begin{proposition}
	Consider $n$ independent random variables $X_1,\ldots,X_n$ and probabilities $p_1,\ldots,p_n$ with $\sum p_i=1$. Consider an independent random variable $R$ that equals $i$ with probability $p_i$. Construct a random variable $Z$ that equals $X_i$ when $R=i$. If each $X_i$ has $\text{VMR}\ge 1$, then $Z$ has $\text{VMR}\ge 1$.
\end{proposition}
\begin{proof}
	We only need to prove this for $n=2$. The case for general $n$ can be proved by mixing two variables iteratively. 
	
	Consider random variables $X,Y$ and construct $Z$ that equals $X$ or $Y$ with probability $p$ or $1-p$. Since $\text{VMR}(X)\ge 1$, $\text{VMR}(Y)\ge 1$, we have $\mathbb{E}(X^2)-[\mathbb{E}(X)]^2\ge \mathbb{E}(X)$ and $\mathbb{E}(Y^2)-[\mathbb{E}(Y)]^2\ge \mathbb{E}(Y)$. Then 
	\begin{equation*}
		\begin{split}
			&\text{VMR}(Z)\\
			=&\frac{p\mathbb{E}(X^2)+(1-p)\mathbb{E}(Y^2)}{p\mathbb{E}(X)+(1-p)\mathbb{E}(Y)}\\
			&+\frac{-p^2[\mathbb{E}(X)]^2-2p(1-p)\mathbb{E}(X)\mathbb{E}(Y)-(1-p)^2[\mathbb{E}(Y)]^2}{p\mathbb{E}(X)+(1-p)\mathbb{E}(Y)}\\
			=&\frac{p\mathbb{E}(X^2)-p[\mathbb{E}(X)]^2+(1-p)\mathbb{E}(Y^2)-(1-p)[\mathbb{E}(Y)]^2}{p\mathbb{E}(X)+(1-p)\mathbb{E}(Y)}\\
			&+\frac{p(1-p)[\mathbb{E}(X)]^2-2p(1-p)\mathbb{E}(X)\mathbb{E}(Y)+p(1-p)[\mathbb{E}(Y)]^2}{p\mathbb{E}(X)+(1-p)\mathbb{E}(Y)}\\
			\ge & \frac{p\mathbb{E}(X)+(1-p)\mathbb{E}(Y)}{p\mathbb{E}(X)+(1-p)\mathbb{E}(Y)}+\frac{p(1-p)[\mathbb{E}(X)-\mathbb{E}(Y)]^2}{p\mathbb{E}(X)+(1-p)\mathbb{E}(Y)}\\
			\ge & 1.
		\end{split}
	\end{equation*}	
\end{proof}




\bibliographystyle{acm}
\bibliography{AC}

\begin{thebibliography}{10}

\bibitem{angelini2022model}
{\sc Angelini, E., Wang, Y., Zhou, J.~X., Qian, H., and Huang, S.}
\newblock A model for the intrinsic limit of cancer therapy: {Duality} of
  treatment-induced cell death and treatment-induced stemness.
\newblock {\em PLOS Computational Biology 18}, 7 (2022), e1010319.

\bibitem{barros2011cdx2}
{\sc Barros, R., da~Costa, L.~T., Pinto-de Sousa, J., Duluc, I., Freund, J.-N.,
  David, L., and Almeida, R.}
\newblock {CDX2} autoregulation in human intestinal metaplasia of the stomach:
  impact on the stability of the phenotype.
\newblock {\em Gut 60}, 3 (2011), 290--298.

\bibitem{baudrimont2019contribution}
{\sc Baudrimont, A., Jaquet, V., Wallerich, S., Voegeli, S., and Becskei, A.}
\newblock Contribution of {RNA} degradation to intrinsic and extrinsic noise in
  gene expression.
\newblock {\em Cell Rep. 26}, 13 (2019), 3752--3761.

\bibitem{baumdick2018conformational}
{\sc Baumdick, M., Gell{\'e}ri, M., Uttamapinant, C., Ber{\'a}nek, V., Chin,
  J.~W., and Bastiaens, P.~I.}
\newblock A conformational sensor based on genetic code expansion reveals an
  autocatalytic component in {EGFR} activation.
\newblock {\em Nat. Commun. 9}, 1 (2018), 1--13.

\bibitem{bokes2012multiscale}
{\sc Bokes, P., King, J.~R., Wood, A.~T., and Loose, M.}
\newblock Multiscale stochastic modelling of gene expression.
\newblock {\em J. Math. Biol. 65}, 3 (2012), 493--520.

\bibitem{bouuaert2013autoregulation}
{\sc Bouuaert, C.~C., Lipkow, K., Andrews, S.~S., Liu, D., and Chalmers, R.}
\newblock The autoregulation of a eukaryotic {DNA} transposon.
\newblock {\em eLife 2\/} (2013).

\bibitem{braichenko2021distinguishing}
{\sc Braichenko, S., Holehouse, J., and Grima, R.}
\newblock Distinguishing between models of mammalian gene expression:
  telegraph-like models versus mechanistic models.
\newblock {\em J. R. Soc. Interface 18}, 183 (2021), 20210510.

\bibitem{cagnetta2019noncanonical}
{\sc Cagnetta, R., Wong, H. H.-W., Frese, C.~K., Mallucci, G.~R., Krijgsveld,
  J., and Holt, C.~E.}
\newblock Noncanonical modulation of the {eIF2} pathway controls an increase in
  local translation during neural wiring.
\newblock {\em Mol. Cell 73}, 3 (2019), 474--489.

\bibitem{cao2018linear}
{\sc Cao, Z., and Grima, R.}
\newblock Linear mapping approximation of gene regulatory networks with
  stochastic dynamics.
\newblock {\em Nature communications 9}, 1 (2018), 1--15.

\bibitem{cao2020analytical}
{\sc Cao, Z., and Grima, R.}
\newblock Analytical distributions for detailed models of stochastic gene
  expression in eukaryotic cells.
\newblock {\em Proc. Natl. Acad. Sci. U.S.A. 117}, 9 (2020), 4682--4692.

\bibitem{carrier1999investigating}
{\sc Carrier, T.~A., and Keasling, J.~D.}
\newblock Investigating autocatalytic gene expression systems through
  mechanistic modeling.
\newblock {\em J. Theor. Biol. 201}, 1 (1999), 25--36.

\bibitem{chahar2014chromatin}
{\sc Chahar, S., Gandhi, V., Yu, S., Desai, K.,
  Cowper-Sal{\textperiodcentered}lari, R., Kim, Y., Perekatt, A.~O., Kumar, N.,
  Thackray, J.~K., Musolf, A., et~al.}
\newblock Chromatin profiling reveals regulatory network shifts and a
  protective role for hepatocyte nuclear factor 4$\alpha$ during colitis.
\newblock {\em Mol. Cell. Biol. 34}, 17 (2014), 3291--3304.

\bibitem{chan2017gene}
{\sc Chan, T.~E., Stumpf, M.~P., and Babtie, A.~C.}
\newblock Gene regulatory network inference from single-cell data using
  multivariate information measures.
\newblock {\em Cell Syst. 5}, 3 (2017), 251--267.

\bibitem{chen2020limit}
{\sc Chen, X., and Jia, C.}
\newblock Limit theorems for generalized density-dependent {Markov} chains and
  bursty stochastic gene regulatory networks.
\newblock {\em J. Math. Biol. 80}, 4 (2020), 959--994.

\bibitem{chen2016overshoot}
{\sc Chen, X., Wang, Y., Feng, T., Yi, M., Zhang, X., and Zhou, D.}
\newblock The overshoot and phenotypic equilibrium in characterizing cancer
  dynamics of reversible phenotypic plasticity.
\newblock {\em Journal of Theoretical Biology 390\/} (2016), 40--49.

\bibitem{cunningham2015mechanisms}
{\sc Cunningham, T.~J., and Duester, G.}
\newblock Mechanisms of retinoic acid signalling and its roles in organ and
  limb development.
\newblock {\em Nat. Rev. Mol. Cell. Biol. 16}, 2 (2015), 110--123.

\bibitem{czuppon2018limits}
{\sc Czuppon, P., and Pfaffelhuber, P.}
\newblock Limits of noise for autoregulated gene expression.
\newblock {\em J. Math. Biol. 77}, 4 (2018), 1153--1191.

\bibitem{dessalles2017stochastic}
{\sc Dessalles, R., Fromion, V., and Robert, P.}
\newblock A stochastic analysis of autoregulation of gene expression.
\newblock {\em J. Math. Biol. 75}, 5 (2017), 1253--1283.

\bibitem{dobrinic2021prc1}
{\sc Dobrini{\'c}, P., Szczurek, A.~T., and Klose, R.~J.}
\newblock {PRC1} drives {Polycomb-mediated} gene repression by controlling
  transcription initiation and burst frequency.
\newblock {\em Nat. Struct. Mol. Biol. 28}, 10 (2021), 811--824.

\bibitem{eden2010drawing}
{\sc Eden, U.~T., and Kramer, M.~A.}
\newblock Drawing inferences from {Fano} factor calculations.
\newblock {\em J. Neurosci. Methods 190}, 1 (2010), 149--152.

\bibitem{fang2017sirt7}
{\sc Fang, J., Ianni, A., Smolka, C., Vakhrusheva, O., Nolte, H., Kr{\"u}ger,
  M., Wietelmann, A., Simonet, N.~G., Adrian-Segarra, J.~M., Vaquero, A.,
  et~al.}
\newblock {Sirt7} promotes adipogenesis in the mouse by inhibiting
  autocatalytic activation of {Sirt1}.
\newblock {\em Proc. Natl. Acad. Sci. U.S.A. 114}, 40 (2017), E8352--E8361.

\bibitem{feigelman2016analysis}
{\sc Feigelman, J., Ganscha, S., Hastreiter, S., Schwarzfischer, M., Filipczyk,
  A., Schroeder, T., Theis, F.~J., Marr, C., and Claassen, M.}
\newblock Analysis of cell lineage trees by exact bayesian inference identifies
  negative autoregulation of {Nanog} in mouse embryonic stem cells.
\newblock {\em Cell Syst. 3}, 5 (2016), 480--490.

\bibitem{firman2018maximum}
{\sc Firman, T., Wedekind, S., McMorrow, T., and Ghosh, K.}
\newblock Maximum caliber can characterize genetic switches with multiple
  hidden species.
\newblock {\em J. Phys. Chem. B 122}, 21 (2018), 5666--5677.

\bibitem{giovanini2020comparative}
{\sc Giovanini, G., Sabino, A.~U., Barros, L.~R., and Ramos, A.~F.}
\newblock A comparative analysis of noise properties of stochastic binary
  models for a self-repressing and for an externally regulating gene.
\newblock {\em Math. Biosci. Eng. 17}, 5 (2020), 5477--5503.

\bibitem{gronlund2013transcription}
{\sc Gr{\"o}nlund, A., L{\"o}tstedt, P., and Elf, J.}
\newblock Transcription factor binding kinetics constrain noise suppression via
  negative feedback.
\newblock {\em Nat. Commun. 4}, 1 (2013), 1--5.

\bibitem{guo2010resolution}
{\sc Guo, G., Huss, M., Tong, G.~Q., Wang, C., Sun, L.~L., Clarke, N.~D., and
  Robson, P.}
\newblock Resolution of cell fate decisions revealed by single-cell gene
  expression analysis from zygote to blastocyst.
\newblock {\em Dev. Cell 18}, 4 (2010), 675--685.

\bibitem{ham2020extrinsic}
{\sc Ham, L., Brackston, R.~D., and Stumpf, M.~P.}
\newblock Extrinsic noise and heavy-tailed laws in gene expression.
\newblock {\em Phys. Rev. Lett. 124}, 10 (2020), 108101.

\bibitem{hara2006cytokinesis}
{\sc Hara, T., Abe, M., Inoue, H., Yu, L., Veenstra, T.~D., Kang, Y., Lee, K.,
  and Miki, T.}
\newblock Cytokinesis regulator {ECT2} changes its conformation through
  phosphorylation at {Thr-341} in {G2/M} phase.
\newblock {\em Oncogene 25}, 4 (2006), 566--578.

\bibitem{hilfinger2016constraints}
{\sc Hilfinger, A., Norman, T.~M., Vinnicombe, G., and Paulsson, J.}
\newblock Constraints on fluctuations in sparsely characterized biological
  systems.
\newblock {\em Phys Rev. Lett. 116}, 5 (2016), 058101.

\bibitem{hornos2005self}
{\sc Hornos, J.~E., Schultz, D., Innocentini, G.~C., Wang, J., Walczak, A.~M.,
  Onuchic, J.~N., and Wolynes, P.~G.}
\newblock Self-regulating gene: an exact solution.
\newblock {\em Phys. Rev. E 72}, 5 (2005), 051907.

\bibitem{hui2020increased}
{\sc Hui, Z., Jiang, Z., Qiao, D., Bo, Z., Qiyuan, K., Shaohua, B., Wenbing,
  Y., Wei, L., Cheng, L., Shuangning, L., et~al.}
\newblock Increased expression of {LCN2} formed a positive feedback loop with
  activation of the {ERK} pathway in human kidney cells during kidney stone
  formation.
\newblock {\em Sci. Rep. 10}, 1 (2020), 1--12.

\bibitem{jia2017simplification}
{\sc Jia, C.}
\newblock Simplification of {Markov} chains with infinite state space and the
  mathematical theory of random gene expression bursts.
\newblock {\em Phys. Rev. E 96}, 3 (2017), 032402.

\bibitem{jia2020kinetic}
{\sc Jia, C.}
\newblock Kinetic foundation of the zero-inflated negative binomial model for
  single-cell {RNA} sequencing data.
\newblock {\em SIAM J. Appl. Math. 80}, 3 (2020), 1336--1355.

\bibitem{jia2020dynamical}
{\sc Jia, C., and Grima, R.}
\newblock Dynamical phase diagram of an auto-regulating gene in fast switching
  conditions.
\newblock {\em J. Chem. Phys. 152}, 17 (2020), 174110.

\bibitem{jia2020small}
{\sc Jia, C., and Grima, R.}
\newblock Small protein number effects in stochastic models of autoregulated
  bursty gene expression.
\newblock {\em J. Chem. Phys. 152}, 8 (2020), 084115.

\bibitem{jia2021frequency}
{\sc Jia, C., and Grima, R.}
\newblock Frequency domain analysis of fluctuations of {mRNA} and protein copy
  numbers within a cell lineage: theory and experimental validation.
\newblock {\em Phys. Rev. X 11}, 2 (2021), 021032.

\bibitem{jia2018relaxation}
{\sc Jia, C., Qian, H., Chen, M., and Zhang, M.~Q.}
\newblock Relaxation rates of gene expression kinetics reveal the feedback
  signs of autoregulatory gene networks.
\newblock {\em The journal of Chemical physics 148}, 9 (2018), 095102.

\bibitem{jia2017stochastic}
{\sc Jia, C., Xie, P., Chen, M., and Zhang, M.~Q.}
\newblock Stochastic fluctuations can reveal the feedback signs of gene
  regulatory networks at the single-molecule level.
\newblock {\em Sci. Rep. 7}, 1 (2017), 1--9.

\bibitem{jia2017emergent}
{\sc Jia, C., Zhang, M.~Q., and Qian, H.}
\newblock Emergent {L{\'e}vy} behavior in single-cell stochastic gene
  expression.
\newblock {\em Phys. Rev. E 96}, 4 (2017), 040402.

\bibitem{jiang2017phenotypic}
{\sc Jiang, D.-Q., Wang, Y., and Zhou, D.}
\newblock Phenotypic equilibrium as probabilistic convergence in
  multi-phenotype cell population dynamics.
\newblock {\em PLOS ONE 12}, 2 (2017), e0170916.

\bibitem{kang2015flexibility}
{\sc Kang, Y., Gu, C., Yuan, L., Wang, Y., Zhu, Y., Li, X., Luo, Q., Xiao, J.,
  Jiang, D., Qian, M., et~al.}
\newblock Flexibility and symmetry of prokaryotic genome rearrangement reveal
  lineage-associated core-gene-defined genome organizational frameworks (vol 5,
  e01867, 2014).
\newblock {\em MBIO 6}, 1 (2015).

\bibitem{karamyshev2018lost}
{\sc Karamyshev, A.~L., and Karamysheva, Z.~N.}
\newblock Lost in translation: ribosome-associated {mRNA} and protein quality
  controls.
\newblock {\em Front. Genet. 9\/} (2018), 431.

\bibitem{karmakar2021effect}
{\sc Karmakar, R., and Das, A.~K.}
\newblock Effect of transcription reinitiation in stochastic gene expression.
\newblock {\em J. Stat. Mech. Theory Exp. 2021}, 3 (2021), 033502.

\bibitem{kidder2010examination}
{\sc Kidder, B.~L., and Palmer, S.}
\newblock Examination of transcriptional networks reveals an important role for
  {TCFAP2C, SMARCA4, and EOMES} in trophoblast stem cell maintenance.
\newblock {\em Genome Res. 20}, 4 (2010), 458--472.

\bibitem{ko2019markov}
{\sc Ko, Y., Kim, J., and Rodriguez-Zas, S.~L.}
\newblock {Markov chain Monte Carlo} simulation of a {Bayesian} mixture model
  for gene network inference.
\newblock {\em Genes Genom. 41}, 5 (2019), 547--555.

\bibitem{levenberg1998long}
{\sc Levenberg, S., Katz, B.-Z., Yamada, K.~M., and Geiger, B.}
\newblock Long-range and selective autoregulation of cell-cell or cell-matrix
  adhesions by cadherin or integrin ligands.
\newblock {\em J. Cell Sci. 111}, 3 (1998), 347--357.

\bibitem{lin2021disentangling}
{\sc Lin, J., and Amir, A.}
\newblock Disentangling intrinsic and extrinsic gene expression noise in
  growing cells.
\newblock {\em Phys. Rev. Lett. 126}, 7 (2021), 078101.

\bibitem{luecken2019current}
{\sc Luecken, M.~D., and Theis, F.~J.}
\newblock Current best practices in single-cell {RNA-seq} analysis: a tutorial.
\newblock {\em Mol. Syst. Biol. 15}, 6 (2019), e8746.

\bibitem{mahajan2021topological}
{\sc Mahajan, T., Singh, A., and Dar, R.}
\newblock Topological constraints on noise propagation in gene regulatory
  networks.
\newblock {\em bioRxiv\/} (2021).

\bibitem{moignard2015decoding}
{\sc Moignard, V., Woodhouse, S., Haghverdi, L., Lilly, A.~J., Tanaka, Y.,
  Wilkinson, A.~C., Buettner, F., Macaulay, I.~C., Jawaid, W., Diamanti, E.,
  et~al.}
\newblock Decoding the regulatory network of early blood development from
  single-cell gene expression measurements.
\newblock {\em Nat. Biotechnol. 33}, 3 (2015), 269--276.

\bibitem{munsky2012using}
{\sc Munsky, B., Neuert, G., and Van~Oudenaarden, A.}
\newblock Using gene expression noise to understand gene regulation.
\newblock {\em Science 336}, 6078 (2012), 183--187.

\bibitem{niu2015phenotypic}
{\sc Niu, Y., Wang, Y., and Zhou, D.}
\newblock The phenotypic equilibrium of cancer cells: {From} average-level
  stability to path-wise convergence.
\newblock {\em Journal of Theoretical Biology 386\/} (2015), 7--17.

\bibitem{norris1998markov}
{\sc Norris, J.~R.}
\newblock {\em Markov chains}.
\newblock Cambridge university press, 1998.

\bibitem{paulsson2005models}
{\sc Paulsson, J.}
\newblock Models of stochastic gene expression.
\newblock {\em Phys. Life Rev. 2}, 2 (2005), 157--175.

\bibitem{pramono2016thrombopoietin}
{\sc Pramono, A., Zahabi, A., Morishima, T., Lan, D., Welte, K., and Skokowa,
  J.}
\newblock Thrombopoietin induces hematopoiesis from mouse {ES} cells via
  {HIF-1}$\alpha$--dependent activation of a {BMP4} autoregulatory loop.
\newblock {\em Ann. N. Y. Acad. Sci. 1375}, 1 (2016), 38--51.

\bibitem{psaila2016single}
{\sc Psaila, B., Barkas, N., Iskander, D., Roy, A., Anderson, S., Ashley, N.,
  Caputo, V.~S., Lichtenberg, J., Loaiza, S., Bodine, D.~M., et~al.}
\newblock Single-cell profiling of human megakaryocyte-erythroid progenitors
  identifies distinct megakaryocyte and erythroid differentiation pathways.
\newblock {\em Genome Biol. 17}, 1 (2016), 1--19.

\bibitem{ramos2015gene}
{\sc Ramos, A.~F., Hornos, J. E.~M., and Reinitz, J.}
\newblock Gene regulation and noise reduction by coupling of stochastic
  processes.
\newblock {\em Phys. Rev. E 91}, 2 (2015), 020701.

\bibitem{sachs2005causal}
{\sc Sachs, K., Perez, O., Pe'er, D., Lauffenburger, D.~A., and Nolan, G.~P.}
\newblock Causal protein-signaling networks derived from multiparameter
  single-cell data.
\newblock {\em Science 308}, 5721 (2005), 523--529.

\bibitem{sanchez2018bayesian}
{\sc Sanchez-Castillo, M., Blanco, D., Tienda-Luna, I.~M., Carrion, M., and
  Huang, Y.}
\newblock A bayesian framework for the inference of gene regulatory networks
  from time and pseudo-time series data.
\newblock {\em Bioinformatics 34}, 6 (2018), 964--970.

\bibitem{shahrezaei2008analytical}
{\sc Shahrezaei, V., and Swain, P.~S.}
\newblock Analytical distributions for stochastic gene expression.
\newblock {\em Proc. Natl. Acad. Sci. U.S.A. 105}, 45 (2008), 17256--17261.

\bibitem{sharma2014markov}
{\sc Sharma, A., and Adlakha, N.}
\newblock Markov chain model to study the gene expression.
\newblock {\em Adv. Appl. Sci. Res. 5}, 2 (2014), 387--393.

\bibitem{shen2019distributed}
{\sc Shen, H., Huo, S., Yan, H., Park, J.~H., and Sreeram, V.}
\newblock Distributed dissipative state estimation for {Markov} jump genetic
  regulatory networks subject to round-robin scheduling.
\newblock {\em IEEE Trans. Neural Netw. Learn. Syst. 31}, 3 (2019), 762--771.

\bibitem{shen2002network}
{\sc Shen-Orr, S.~S., Milo, R., Mangan, S., and Alon, U.}
\newblock Network motifs in the transcriptional regulation network of
  escherichia coli.
\newblock {\em Nat. Genet. 31}, 1 (2002), 64--68.

\bibitem{sheth2014self}
{\sc Sheth, R., Bastida, M.~F., Kmita, M., and Ros, M.}
\newblock {``Self-regulation,''} a new facet of {Hox} genes' function.
\newblock {\em Dev. Dyn. 243}, 1 (2014), 182--191.

\bibitem{shmulevich2003steady}
{\sc Shmulevich, I., Gluhovsky, I., Hashimoto, R.~F., Dougherty, E.~R., and
  Zhang, W.}
\newblock Steady-state analysis of genetic regulatory networks modelled by
  probabilistic boolean networks.
\newblock {\em Comp. Funct. Genomics 4}, 6 (2003), 601--608.

\bibitem{skinner2016single}
{\sc Skinner, S.~O., Xu, H., Nagarkar-Jaiswal, S., Freire, P.~R., Zwaka, T.~P.,
  and Golding, I.}
\newblock Single-cell analysis of transcription kinetics across the cell cycle.
\newblock {\em Elife 5\/} (2016), e12175.

\bibitem{stewart2013under}
{\sc Stewart, A.~J., Seymour, R.~M., Pomiankowski, A., and Reuter, M.}
\newblock Under-dominance constrains the evolution of negative autoregulation
  in diploids.
\newblock {\em PLOS Comput. Biol. 9}, 3 (2013), e1002992.

\bibitem{swain2004efficient}
{\sc Swain, P.~S.}
\newblock Efficient attenuation of stochasticity in gene expression through
  post-transcriptional control.
\newblock {\em J. Mol. Biol. 344}, 4 (2004), 965--976.

\bibitem{swain2002intrinsic}
{\sc Swain, P.~S., Elowitz, M.~B., and Siggia, E.~D.}
\newblock Intrinsic and extrinsic contributions to stochasticity in gene
  expression.
\newblock {\em Proc. Natl. Acad. Sci. U.S.A. 99}, 20 (2002), 12795--12800.

\bibitem{thattai2001intrinsic}
{\sc Thattai, M., and Van~Oudenaarden, A.}
\newblock Intrinsic noise in gene regulatory networks.
\newblock {\em Proc. Natl. Acad. Sci. U.S.A. 98}, 15 (2001), 8614--8619.

\bibitem{thomas2019intrinsic}
{\sc Thomas, P.}
\newblock Intrinsic and extrinsic noise of gene expression in lineage trees.
\newblock {\em Sci. Rep. 9}, 1 (2019), 1--16.

\bibitem{thomas2013reliable}
{\sc Thomas, P., Matuschek, H., and Grima, R.}
\newblock How reliable is the linear noise approximation of gene regulatory
  networks?
\newblock {\em BMC Genom. 14}, 4 (2013), 1--15.

\bibitem{veerman2021parameter}
{\sc Veerman, F., Popovi{\'c}, N., and Marr, C.}
\newblock Parameter inference with analytical propagators for stochastic models
  of autoregulated gene expression.
\newblock {\em Int. J. Nonlinear Sci. Numer. Simul.\/} (2021).

\bibitem{voliotis2008fluctuations}
{\sc Voliotis, M., Cohen, N., Molina-Par{\'\i}s, C., and Liverpool, T.~B.}
\newblock Fluctuations, pauses, and backtracking in {DNA} transcription.
\newblock {\em Biophys. J. 94}, 2 (2008), 334--348.

\bibitem{wang2019roles}
{\sc Wang, D.-G., Wang, S., Huang, B., and Liu, F.}
\newblock Roles of cellular heterogeneity, intrinsic and extrinsic noise in
  variability of p53 oscillation.
\newblock {\em Sci. Rep. 9}, 1 (2019), 1--11.

\bibitem{wang2018some}
{\sc Wang, Y.}
\newblock {\em Some Problems in Stochastic Dynamics and Statistical Analysis of
  Single-Cell Biology of Cancer}.
\newblock {Ph.D. thesis}, University of Washington, 2018.

\bibitem{wang2022impossibility}
{\sc Wang, Y.}
\newblock Impossibility results about inheritance and order of death.
\newblock {\em PLOS ONE 17}, 11 (2022), e0277430.

\bibitem{wang2022two}
{\sc Wang, Y.}
\newblock Two metrics on rooted unordered trees with labels.
\newblock {\em Algorithms for Molecular Biology 17}, 1 (2022), 1--17.

\bibitem{wang2023longest}
{\sc Wang, Y.}
\newblock Longest common subsequence algorithms and applications in determining
  transposable genes.
\newblock {\em arXiv preprint arXiv:2301.03827\/} (2023).

\bibitem{wang2022modelling}
{\sc Wang, Y., Dessalles, R., and Chou, T.}
\newblock Modelling the impact of birth control policies on {China}'s
  population and age: effects of delayed births and minimum birth age
  constraints.
\newblock {\em Royal Society Open Science 9}, 6 (2022), 211619.

\bibitem{wang2020biological}
{\sc Wang, Y., Kropp, J., and Morozova, N.}
\newblock Biological notion of positional information/value in morphogenesis
  theory.
\newblock {\em International Journal of Developmental Biology 64}, 10-11-12
  (2020), 453--463.

\bibitem{wang2020model}
{\sc Wang, Y., Minarsky, A., Penner, R., Soul{\'e}, C., and Morozova, N.}
\newblock Model of morphogenesis.
\newblock {\em Journal of Computational Biology 27}, 9 (2020), 1373--1383.

\bibitem{wang2022discrete}
{\sc Wang, Y., Mistry, B.~A., and Chou, T.}
\newblock Discrete stochastic models of selex: Aptamer capture probabilities
  and protocol optimization.
\newblock {\em The Journal of Chemical Physics 156}, 24 (2022), 244103.

\bibitem{wang2020mathematical}
{\sc Wang, Y., and Qian, H.}
\newblock Mathematical representation of {Clausius'} and {Kelvin's} statements
  of the second law and irreversibility.
\newblock {\em Journal of Statistical Physics 179}, 3 (2020), 808--837.

\bibitem{wang2020causal}
{\sc Wang, Y., and Wang, L.}
\newblock Causal inference in degenerate systems: {An} impossibility result.
\newblock In {\em International Conference on Artificial Intelligence and
  Statistics\/} (2020), PMLR, pp.~3383--3392.

\bibitem{wang2022inference}
{\sc Wang, Y., and Wang, Z.}
\newblock Inference on the structure of gene regulatory networks.
\newblock {\em Journal of Theoretical Biology 539\/} (2022), 111055.

\bibitem{wang2021inference}
{\sc Wang, Y., Zhang, B., Kropp, J., and Morozova, N.}
\newblock Inference on tissue transplantation experiments.
\newblock {\em Journal of Theoretical Biology 520\/} (2021), 110645.

\bibitem{wang2021measuring}
{\sc Wang, Y., and Zheng, Z.}
\newblock Measuring policy performance in online pricing with offline data.
\newblock {\em Available at SSRN 3729003\/} (2021).

\bibitem{wang2023online}
{\sc Wang, Y., Zheng, Z., and Shen, Z.-J.~M.}
\newblock Online pricing with polluted offline data.
\newblock {\em Available at SSRN 4320324\/} (2023).

\bibitem{wang2023multiple}
{\sc Wang, Y., Zhou, J.~X., Pedrini, E., Rubin, I., Khalil, M., Qian, H., and
  Huang, S.}
\newblock {Multiple phenotypes in HL60 leukemia cell population}.
\newblock {\em arXiv preprint arXiv:2301.03782\/} (2023).

\bibitem{werhli2006comparative}
{\sc Werhli, A.~V., Grzegorczyk, M., and Husmeier, D.}
\newblock Comparative evaluation of reverse engineering gene regulatory
  networks with relevance networks, graphical gaussian models and bayesian
  networks.
\newblock {\em Bioinformatics 22}, 20 (2006), 2523--2531.

\bibitem{xing2005causal}
{\sc Xing, B., and Van Der~Laan, M.~J.}
\newblock A causal inference approach for constructing transcriptional
  regulatory networks.
\newblock {\em Bioinformatics 21}, 21 (2005), 4007--4013.

\bibitem{yan2019kinetic}
{\sc Yan, J., Hilfinger, A., Vinnicombe, G., Paulsson, J., et~al.}
\newblock Kinetic uncertainty relations for the control of stochastic reaction
  networks.
\newblock {\em Phys. Rev. Lett. 123}, 10 (2019), 108101.

\bibitem{yang2020machine}
{\sc Yang, W., Peng, L., Zhu, Y., and Hong, L.}
\newblock When machine learning meets multiscale modeling in chemical
  reactions.
\newblock {\em J. Chem. Phys. 153}, 9 (2020), 094117.

\bibitem{ye2016stochastic}
{\sc Ye, F. X.-F., Wang, Y., and Qian, H.}
\newblock Stochastic dynamics: {Markov} chains and random transformations.
\newblock {\em Discrete \& Continuous Dynamical Systems-B 21}, 7 (2016), 2337.

\bibitem{zhou2014multi}
{\sc Zhou, D., Wang, Y., and Wu, B.}
\newblock A multi-phenotypic cancer model with cell plasticity.
\newblock {\em Journal of Theoretical Biology 357\/} (2014), 35--45.

\bibitem{zhou2012analytical}
{\sc Zhou, T., and Zhang, J.}
\newblock Analytical results for a multistate gene model.
\newblock {\em SIAM J. Appl. Math. 72}, 3 (2012), 789--818.

\end{thebibliography}


\end{document}